\documentclass[final,3p,fleqn,12pt]{elsarticle}

\makeatletter
\def\ps@pprintTitle{%
 \let\@oddhead\@empty
 \let\@evenhead\@empty
 \def\@oddfoot{}%
 \let\@evenfoot\@oddfoot}
\makeatother

\usepackage{amssymb}
\usepackage{multicol}
\usepackage{amsmath}
\usepackage{upgreek}

\usepackage{amsmath}
\usepackage{setspace}
\usepackage{amsfonts}
\usepackage{amsthm}
\usepackage{bbm}
\usepackage{mathtools}
\usepackage{subfig}
\usepackage{lineno}
\usepackage{algorithm}
\usepackage{algpseudocode}
\usepackage{hyperref}
\usepackage{float}
\usepackage{array,multirow}
\usepackage{color}
\usepackage{tikz}
\usepackage{graphicx}
\usepackage[flushleft]{threeparttable}
\usepackage{blindtext}
\usepackage[colorinlistoftodos,prependcaption,textsize=tiny]{todonotes}
\usepackage{regexpatch}
\usepackage[export]{adjustbox}
\usepackage{rotating}
\usepackage{amssymb,mleftright}
\usepackage{bm}
\usepackage{caption}
\usepackage{subfig}
\usepackage{booktabs}
\usepackage{lscape}

\definecolor{lightgray}{gray}{0.80}

\usetikzlibrary{decorations.pathreplacing}
\usepackage[listings,theorems,skins,breakable]{tcolorbox}
\tcbset{skin=enhanced}
\newtcolorbox{lbracebox}[1][Word]{%
   frame hidden,enlarge left by=2cm,width=\linewidth-2cm,%
  overlay unbroken = {\draw [decorate,decoration={brace,amplitude=10pt},]%
                     (frame.south west)-- (frame.north west)
                    node [black,midway,left,xshift=-.6cm] {#1};},%
}

\usepackage{scalerel,amssymb}

\makeatletter
\xpatchcmd{\@todo}{\setkeys{todonotes}{#1}}{\setkeys{todonotes}{inline,#1}}{}{}
\makeatother

\theoremstyle{plain}
\newtheorem{theorem}{Theorem}[section]

\newtheorem{lemma}[theorem]{Lemma}

\newtheorem{remark}[theorem]{Remark}

\theoremstyle{definition}

\newcommand{\mat}[1]{\mathbf{#1}}

\usepackage{pifont}
\usepackage{xcolor}
\newcommand{\cmark}{\textcolor{green!80!black}{\ding{51}}}
\newcommand{\xmark}{\textcolor{red}{\ding{55}}}

\begin{document}

\begin{frontmatter}



\title{An Entropy-Stable/Double-Flux scheme for the multi-component compressible Navier-Stokes equations}

\author[TUD,TUD1]{Vahid Badrkhani}
\ead{badrkhani@stfs.tu-darmstadt.de}

\author[TUD]{T. Jeremy P. Karpowski}
\ead{karpowski@stfs.tu-darmstadt.de}

\corref{cor1}
\author[TUD]{Christian Hasse}
\ead{hasse@stfs.tu-darmstadt.de }

\cortext[cor1]{Corresponding author}

\address[TUD]{Institute for Simulation of Reactive Thermo-Fluid Systems, Technical University of Darmstadt, Germany}
\address[TUD1]{Institute for Mechanics, Computational Mechanics Group, Technical University of Darmstadt, Germany}

\begin{abstract} 
We present a novel combination of numerical techniques to improve the efficiency, accuracy, and robustness of multi-component compressible flow simulations. At the core of our approach is an Entropy-Stable formulation that preserves kinetic energy and integrates a Double-Flux scheme tailored for multi-component flows with variable specific heat ratios. This formulation yields low-dissipation, oscillation-free solutions and enhances stability compared to standard fully conservative methods. To further improve robustness, we introduce a new hybrid dissipation strategy that blends the Entropy-Stable/Double-Flux approach with conventional dissipation mechanisms. We provide a rigorous proof that the resulting numerical flux satisfies a semi-discrete entropy inequality, ensuring consistency with the second law of thermodynamics. For time integration, we employ an explicit Runge–Kutta scheme in combination with adaptive mesh refinement to capture local flow features dynamically. The method is implemented within an existing compressible Navier–Stokes solver based on OpenFOAM. Benchmark cases, including multi-dimensional interface and shock-interface interactions, demonstrate the effectiveness of the proposed framework. The results confirm its favorable stability and robustness, validating the approach as a promising advancement for high-fidelity simulations of supersonic flows.
\end{abstract}

\begin{highlights}
\item We introduce an Entropy-Stable/Double-Flux approach for solving multi-component compressible Navier–Stokes equations.

\item To enhance stability, we have coupled an Entropy-Stable approach with the Double-Flux method. With a suitable choice of numerical flux, the scheme ensures entropy stability and satisfies the second law of thermodynamics in an integral sense.

\item Compared to conventional numerical fluxes, the Entropy-Stable/Double-Flux schemes offer superior stability, robustness, and oscillation-free behavior for multi-component compressible Navier–Stokes simulations.
\item We have developed a novel hybrid dissipation method by blending a standard hybrid dissipation term with the Entropy-Stable/Double-Flux algorithm.

\item The proposed numerical flux is rigorously verified and validated through extensive multi-dimensional benchmark tests and complex multi-component compressible flow problems.

\end{highlights}


\end{frontmatter}



\section{Introduction \label{sec:introduction}}

The accurate and robust simulation of multi-component compressible flows with material interfaces is crucial for numerous engineering applications, such as combustion in propulsion systems and explosive detonation products, as well as scientific studies involving flow instabilities, chemical reactions, and phase transitions. These complex flows often feature nonlinear wave interactions, including shock and rarefaction waves, contact discontinuities, and material interfaces between different fluids. To capture these intricate phenomena, interface-capturing methods have been extensively studied. Various numerical approaches have been proposed, including finite difference methods \cite{latini2007effects,kawai2011high,movahed2013solution}, finite volume schemes \cite{vakilipour2019developing,houim2011low,johnsen2006implementation,xiong2012weno}, and discontinuous Galerkin methods \cite{renac2021entropy,lv2014discontinuous,de2015new,gaburro2024discontinuous}, among others.

Tadmor’s pioneering work on Entropy-Conservative (EC) numerical fluxes \cite{tadmor1987numerical,tadmor2003entropy} laid the foundation for developing schemes that ensure compliance with the entropy condition, thereby aligning with the second law of thermodynamics. These EC fluxes, when combined with dissipation terms using entropy variables, result in numerical schemes that inherently produce entropy correctly. Chandrasekhar \cite{chandrashekar2013kinetic} extended Tadmor’s concept to the compressible Euler equations, introducing an EC flux that also preserves kinetic energy in the sense articulated by Jameson \cite{jameson2008formulation}. This kinetic energy preservation is particularly beneficial in turbulent flow simulations, as highlighted by Subbareddy et al. \cite{subbareddy2009fully}. Building on these ideas, more recent works by Coppola et al. \cite{coppola2019numerically} and Kuya et al. \cite{kuya2018kinetic} have proposed schemes that effectively balance both kinetic energy and entropy conservation, collectively referred to as kinetic energy–and entropy-preserving schemes.

For multi-component compressible Euler equations, Entropy-Stable (ES) schemes have been formulated to respect the entropy condition \cite{gouasmi2020formulation}. However, these schemes can suffer from pressure oscillations in moving interface configurations \cite{gouasmi2020formulation} and may produce nonphysical negative densities or pressures, even at first order, in the next time step \cite{abgrall2001computations}. Numerical experiments have shown that while ES schemes handle shocks and stationary contact discontinuities adequately, they struggle to maintain pressure equilibrium and constant velocity in simulations involving moving interfaces \cite{karni1994multicomponent}. A potential solution lies in the EC/ES schemes for non-conservative hyperbolic systems developed by Castro et al. \cite{castro2013entropy}, although non-conservative approaches come with their own set of challenges \cite{hou1994nonconservative,abgrall2010comment}.

Matheis and Hickel \cite{matheis2018multi} conducted a comparative analysis between a fully conservative (FC) scheme and a quasi-conservative (QC) approach. Their findings revealed significantly higher temperature fields in the QC results compared to the FC scheme. Schmitt et al. \cite{schmitt2010large,schmitt2020large} implemented an energy correction procedure alongside an artificial viscosity method to mitigate nonphysical noise caused by steep density gradients. Terashima and Koshi \cite{terashima2012approach} further refined the approach by substituting the energy conservation equation with a pressure transport equation, successfully conserving pressure equilibrium across contact interfaces. 

Thermodynamic properties in multi-component flows are not constant but depend on temperature and species composition. Differences in the specific heat ratio across the material interface, induce spurious pressure oscillations \cite{lv2014discontinuous}. Noteworthy approach for this challange is the Double-Flux model, originally developed for calorically perfect gases by Abgrall and Karni \cite{abgrall2001computations} and later extended by Billet and Abgrall \cite{billet2003adaptive} for reacting flows. This method, which has been adapted for high-order schemes \cite{houim2011low}, has demonstrated the ability to accurately predict shock speeds, even for very strong shock waves. Ma et al. \cite{ma2017entropy,ma2019numerical} extended the Double-Flux method to real-gas flows and introduced an algorithm to enforce entropy stability in the numerical scheme through flux correction. However, this type of flux is not inherently Entropy-Stable, requiring entropy verification at each time step, followed by correction if necessary \cite{tadmor1986minimum,zhang2012minimum}.  

Entropy-Stable schemes satisfy the second law of thermodynamics and effectively dampen numerical oscillations, yet they introduce spurious oscillations at material interfaces with disparate specific heat ratios. Conversely, the Double-Flux method eliminates such interfacial oscillations in multi-component flows but does not ensure thermodynamic consistency and kinetic energy preservation, rendering it unsuitable for high-speed compressible turbulent flows. Bridging these limitations — while maintaining numerical stability, thermodynamic consistency, and interfacial accuracy — remains an unresolved challenge in computational fluid dynamics.

We propose a hybrid discretization strategy that unifies Entropy-Stable and Double-Flux schemes to advance the simulation of the multi-component compressible Navier–Stokes equations. The framework ensures kinetic energy preservation, thermodynamic consistency, and robust suppression of spurious pressure oscillations at material interfaces.
One key novelty is the evolution of auxiliary variables—specifically, the specific heat ratio and total enthalpy—which effectively dampens pressure artifacts typical of Entropy-Stable formulations while maintaining their theoretical rigor. This modification addresses a critical limitation in existing schemes and enhances overall numerical stability.
To contextualize our contribution, Table~\ref{Tab:introduction} compares relevant properties—kinetic energy conservation, thermodynamic compliance, and interface robustness—across established and proposed approaches, illustrating how our method resolves key trade-offs.
The proposed framework is validated through theoretical analysis and a range of benchmark test cases, demonstrating its accuracy, stability, and effectiveness in complex multi-component flow scenarios.

For clarity, we briefly distinguish between Entropy-Conservative and Entropy-Stable schemes. Entropy-Conservative methods [e.g., Tadmor, 1987] satisfy a discrete analog of thermodynamic entropy conservation for smooth solutions but may fail to control entropy growth at discontinuities. Entropy-Stable schemes, in contrast, enforce a stricter inequality ($dS/dt \geq 0$) to ensure entropy production aligns with the second law of thermodynamics, even for discontinuous flows. While both frameworks preserve fundamental physics, Entropy-Stable methods are inherently dissipative, enabling robustness in under-resolved or shocked regions.

The remainder of this paper is structured as follows: Section \ref{S2} introduces the multi-component compressible Navier-Stokes equations and discusses entropy pairs, symmetrization, and the mapping of entropy variables to conservative variables. In Section \ref{sec:Numerical methods}, we present the governing differential equations and describe the spatial and temporal discretizations used in this study, employing the finite volume method. This section focuses on the Entropy-Conservative/Stable, kinetic energy-preserving, and Double-Flux methods for simulating multi-component flows. In Section \ref{sec:Entropy-Stable/Double-Flux scheme}, we derive an Entropy-Stable/Double-Flux method, as well as a novel combined approach that accurately simulates multi-component flows with shock waves while minimizing oscillations and dissipation. Also, we mathematically proof that this method is Entropy-Stable. Section \ref{sec:Implementation aspects} presents the solver implementation used in this study, along with a discussion on adaptive mesh refinement and the computational setup. Section \ref{sec:Numerical results} presents the results of numerical experiments conducted using our Entropy-Stable/Double-Flux methods, compared with the standard and Double-Flux methods. These experiments include various test cases, ranging from one-dimensional to three-dimensional flow scenarios. Finally, the last section provides concluding remarks.

 \begin{table}
\caption{Table compares different models for compressible flow. The column "KE preserving" indicates whether the numerical flux follows kinetic energy-preserving schemes. The "Law of thermodynamics" column specifies whether the model satisfies the second law of thermodynamics. The column titled "Oscillation-free" indicates that the method is free of oscillations at the multi-component flow near the material interface. The symbol \small{$\color{orange!95!black}\bigstar$} denotes the existence of an algorithm designed to enforce entropy stability through flux correction.} 
\centering
\begin{tabular*}{\textwidth}{l@{\extracolsep{\fill}}lccc} 
\toprule
Model 																  									& \small{KE preserving} 							&\small{2nd Law of thermodynamics} 		& \small{Oscillation-free	}		 	\\\midrule
Tadmor \cite{tadmor1987numerical,tadmor2003entropy}			             &\xmark		    											&  \cmark												&\xmark   								 								 \\
Jameson \cite{jameson2008formulation}		      									&\cmark		    											&  \xmark												&\xmark   																 \\
Chandrashekar \cite{chandrashekar2013kinetic} 									&\cmark		    											&  \cmark												&\xmark   								 								 \\
Coppola et al. \cite{coppola2019numerically}									&\cmark		    											&  \cmark												&\xmark \\
Gouasmi et al. \cite{gouasmi2020formulation} 	 									&\cmark		    											&  \cmark												&\xmark   								 								 \\
 Abgrall and Karni. \cite{abgrall2001computations}															&\xmark		    											&  \xmark												&\cmark   								 								 \\
Ma et al. \cite{ma2017entropy} 			           	 										&\xmark			    										& \small{$\color{orange!95!black}\bigstar$}												&\cmark   								   								 \\
Current	 work         												 										&\cmark				   									& \cmark												&\cmark	   								  				 			 \\\hline
\end{tabular*}
\label{Tab:introduction}%
\end{table}

\section{Governing equations\label{sec:Governing equations}}\label{S2}
The time-dependent, compressible Navier-Stokes equations  for a multi-component gas mixture containing $n$ chemical species can be written in the conservation law of Equation \(\ref{C4}\)

\begin{equation}\label{C4}
    \frac{\partial \mathbf{\mat {u}}}{\partial t} + \nabla \cdot \mathbf{F}(\mathbf{\mat {u}}) -\nabla \cdot \mathbf{G}(\mathbf{\mat {u}}, \nabla \mathbf{\mat {u}}) = \mathbf{S},
\end{equation}
with 
 \begin{align}\label{C5}
 \mathbf{u} = \begin{bmatrix} \rho Y_1 \\ \vdots \\ \rho Y_{n-1} \\ \rho \\ \rho \mathbf{v} \\ \rho E  \end{bmatrix},
  \mathbf{F} = \begin{bmatrix}  \rho Y_1 \mathbf{v} \\ \vdots \\ \rho Y_{n-1} \mathbf{v} \\ \rho \mathbf{v} \\ \rho \mathbf{v} \otimes \mathbf{v} + p \mathbf{I} \\ (\rho E + p) \mathbf{v}  \end{bmatrix},
      \mathbf{G} = \begin{bmatrix} J_1 \\ \vdots \\ J_{n-1} \\0 \\ \mathbf{\tau} \\ \mathbf{\tau} \cdot \mathbf{v} - \mathbf{q} - \sum_{i=1}^n h_i J_i \end{bmatrix},
    \mathbf{S} = \begin{bmatrix} \dot{\omega}_1 \\ \vdots \\ \dot{\omega}_{n-1} \\0 \\ 0 \\ 0  \end{bmatrix}
\end{align}
where  \(\mat {F}(\mat u)\), \(\mathbf{G}(\mathbf{\mat {u}}, \nabla \mathbf{\mat {u}})\) and $\mathbf{S}$ are the inviscid, viscous flux functions and source term vector, respectively. In this context, the conservative state vector with the $n_s$-dimensional ($n_s=d+n+1$, with $d$ the number of geometrical dimensions) is \(\mat u = {(\rho Y_i, \rho, {\rho v}_j, \rho E)}^T\) where \(\rho\), \(v_j\), \(Y_i\), $\dot{\omega}_i$, \(E\), \( p \) represent the density, the velocity in the \(j\)th coordinate direction, mass fraction of species \(i\), mass production rate of chemical species \( i \), specific total energy, and pressure, respectively. The equation of state, total energy, specific heat at constant pressure, and the relation for the sum of species mass fractions are given by:
\begin{equation}\label{rcpcv}
\begin{aligned}
p &= \rho R T \sum_{i=1}^{n} \frac{Y_i}{m_i} = \rho T \sum_{i=1}^{n} Y_i \frac{R}{m_i} =r \rho T\\
    E &= \sum_{i=1}^{n} Y_i \left( e_{0i} +\int_{T_0}^{T} c_{pi}(T) \, dT\right) - \frac{p}{\rho} + \frac{|\mathbf{v}|^2}{2}, \\
    c_p &= \sum_{i=1}^{n} Y_i {c_p}_i, \quad c_v = \sum_{i=1}^{n} Y_i {c_v}_{i}, \quad  r = \sum_{i=1}^{n} Y_i r_i = c_p - c_v
\end{aligned}
\end{equation}
where \( R \) is the universal gas constant, \( T \) is the temperature, \( m_i \) is the molar mass of species \( i \) and \( c_{pi} \) is the specific heat capacity at constant pressure for species \( i \), which is here expressed in polynomial form as a function of temperature \cite{mcbride2002nasa}. In this work, the temperature \( T \) is found by solving \cite{renac2021entropy}:
\begin{equation}\label{C6}
\rho e = \sum \rho Y_i e_i, \quad e_i := e_{0i} + c_{vi} T
\end{equation}
where \( \rho e = \rho E - \rho \mathbf{v}^2 /2\), for each species \( i \), \( e_{0i} \) is a constant, and \( c_{vi} \) represents the specific heat at constant volume. 
The total energy per volume for a calorically perfect gas can be written as
\begin{equation}\label{Ecal}
E_{cal} = e_{0} + \frac{p}{\rho (\gamma-1)} + \frac{|\mathbf{v}|^2}{2}
\end{equation}
where $\gamma = c_p /c_v$ is the ratio of gas specific heats and $e_{0}$ is the reference internal energy at a reference temperature \( T_0 \).

The thermo-viscous-diffusive transport terms, appearing in Equation \(\ref{C5}\), take the following form:
\begin{equation}
\begin{aligned}
    \mathbf{J} &= -\rho \alpha \nabla Y, \\
    \mathbf{\tau} &= -\frac{2}{3} \mu (\nabla \cdot \mathbf{u}) \mathbf{I} + \mu \left( \nabla \mathbf{u} + (\nabla \mathbf{u})^T \right), \\
    \mathbf{q} &= -\kappa \nabla T,
\end{aligned}
\end{equation}
where \( \mathbf{J} \) is the species diffusion flux matrix, \( \mathbf{\tau} \) is the viscous stress tensor, described using Newton's law; and \( \mathbf{q} \) is the heat-flux vector, modeled using Fourier's law. In these equations, \( \alpha \) represents the mixture-averaged species diffusivities, \( \mu \) is the dynamic viscosity, and \( \kappa \) is the thermal conductivity.

\subsection{Entropy pairs and symmetrization of the governing equations}\label{S22}

Nonlinear hyperbolic systems of conservation laws arising from physical systems, such as the compressible Euler equations, commonly admit a generalized entropy pair  $\left( U\left( \mat u\right) , \boldsymbol{\mathcal{F}}\left( \mat u\right) \right) $ consisting of a convex generalized 
entropy function  $U\left( \mat u\right) : \mathbb{R}^{{n}_{s}}\rightarrow  \mathbb{R}$ and an entropy flux  $\boldsymbol{\mathcal{F}}\left( \mat u\right) : \mathbb{R}^{{n}_{s}}\rightarrow  \mathbb{R}^{d}$  that satisfies

 \begin{align}\label{s2200}
 \frac{\partial {\mathcal{F}_{j}}  }{\partial u_{k}} = \frac{\partial  {F}_{ij}}{\partial u_{k}} \frac{\partial  U}{\partial u_{i}},
      \quad   i, k =1, ..., {n}_{s},\quad \quad  j= 1,..., d .
\end{align}
Entropy pairs exist if and only if the hyperbolic system is symmetrized via the change of variables  $\bm{v}\left( \mat u\right) = {\partial U}/{\partial  \mat u}$ \cite{godunov1961interesting,renac2021entropy}, where  $\bm{v}$ are referred to as the entropy variables. An important property of entropy symmetrized hyperbolic systems emerges when the inner product of the conservation law is taken with respect to the entropy variables, namely, the following identities hold for smooth solutions \cite{barth1999numerical}

 \begin{align}\label{s220}
 \bm{v}^{T} \cdot \frac{\partial \mat u \left( \bm{v} \right)  }{\partial t} = \frac{\partial  U}{\partial t},
      \quad   \bm{v}^{T} \cdot \left(  \nabla  \cdot \mat{F} \right) = \nabla  \cdot \boldsymbol{\mathcal{F}}.
\end{align}

We consider the
 \begin{align}\label{s221}
 U = - \rho s(\bm{u}),
      \quad  \boldsymbol{\mathcal{F}}= - \rho s(\bm{u}) \bm{\mathbf{v}},
\end{align}
where thermodynamic entropy of the mixture, \( s \), is given by:

\begin{equation}\label{s}
    s := \sum_{i=1}^{n} Y_i s_i, \quad s_i := c_{vi} \ln(T) - \frac{R}{m_i} \ln(\rho Y_i),
\end{equation}
where \( s \) denotes the specific entropy of the mixture, \( s_i \) and \( c_{vi} \) are the specific entropy and constant volume specific heat for species \( i \), respectively. We note that entropy-satisfying solutions of the Navier-Stokes equation satisfy
 \begin{align}\label{s223}
\dfrac{\partial U}{\partial t} + \nabla  \cdot \boldsymbol{\mathcal{F}} \leq 0
\end{align}

\subsection{Mapping of Entropy Variables to Conservative Variables}\label{S22}
Using Equation \(\ref{s}\) and the Gibbs identity, we can write \cite{renac2021entropy}
\begin{equation}
T \, d(\rho s) = d(\rho e) - \sum_{i=1}^{n} g_i \, d\rho_i = d(\rho E) - \bm{\mathbf{v}} \cdot d(\rho \bm{\mathbf{v}}) + \frac{\bm{\mathbf{v}} \cdot \bm{\mathbf{v}}}{2} d\rho-
\sum_{i=1}^{n} g_i  \, d\rho_i
\end{equation}
where $g_i$ is the Gibbs function of species $i$, so the entropy variables can be written as
\begin{equation}\label{Esv}
 \bm{v}\left( \mat u\right)=\left(\frac{\partial U}{\partial \mat u}\right)= \frac{1}{T}
\begin{bmatrix} 
g_1 - g_n & \cdots & g_{n-1} - g_n & -g_n -\frac{\bm{\mathbf{v}} \cdot \bm{\mathbf{v}}}{2} & \bm{\mathbf{v}} & -1
\end{bmatrix}^T ,
\end{equation}
and the entropy potential flux can be derived:
\begin{equation}
\bm{\psi(u)} = \mathbf{F}(\mathbf{\mat {u}})^T \ \bm{v}\left( \mat u\right) - \boldsymbol{\mathcal{F}} .
\end{equation}
From this expression, using \(\ref{C5}\), \(\ref{s221}\), \(\ref{Esv}\), and the relation $\rho_i (e_i -c_{pi} T)= p$ \cite{renac2021entropy}, we have 

\begin{equation}\label{po1}
\begin{aligned}
\bm{\psi(u)} &=\frac{1}{T}\left( \sum_{i=1}^{n-1} \left(g_i - g_n\right)\rho_i \bm{\mathbf{v}} + \left(-g_n -\frac{\bm{\mathbf{v}} \cdot \bm{\mathbf{v}}}{2}\right)\rho \bm{\mathbf{v}} 
             + \left( \rho \mathbf{v} \otimes \mathbf{v} + p \mathbf{I}\right)\bm{\mathbf{v}} 
             - \left( \rho E + p\right)\bm{\mathbf{v}} \right)\\
             &- \rho s \bm{\mathbf{v}} 
             = \sum_{i=1}^{n} c_{pi} \rho_i \bm{\mathbf{v}} - \frac{e}{T} \rho \bm{\mathbf{v}} = (c_{p} -c_{v})\rho \bm{\mathbf{v}} = r(\mathbf{Y}) \rho \bm{\mathbf{v}} .
\end{aligned}
\end{equation}

\begin{remark} In order to derive the mapping from conservative variables to entropy variables, one can first compute the velocity and  Gibbs functions as:
\begin{equation}\label{gib}
\text{v}_j =\text{u}_{n+j}/ \text{u}_{n}, \quad \quad j=1, \cdots,d ,\\
g_i = h_i -T s_i = e_i + r_i T -T s_i, 
\end{equation}
where $h_i = e_i +r_i T$ is the specific enthalpy of species $i$ and  $r_i = c_{pi} -c_{vi}$. Using these variables and Equation \(\ref{Esv}\), we can obtain entropy variables for multi-component compressible flow.
\end{remark}
\begin{remark} Likewise, we can transform from the entropy variables to conservative variables by
\begin{equation}
\mathbf{u} = 
\begin{bmatrix} 
\rho_i & \cdots & \rho_{n-1} &\rho =\sum_{i=1}^{n} \rho_i & \rho \bm{\mathbf{v}} & \rho c_{v} T + \frac{(\rho \mathbf{v})^2}{2\rho}
\end{bmatrix}^T
\end{equation}
where, in terms of the entropy variables, we have
\begin{equation}
\begin{aligned}
    T&= -\frac{1}{v_{n+d+1}},\quad &\text{v}_j &=  T v_{n+j},\quad &j &= 1, \dots, d\\
    g_n &= T v_n + \frac{\sum_{j=1}^{d} \text{v}_j^2}{2},  &g_i &= T v_i + g_n, \quad &i &= 1, \dots, n-1, \\
    s_i &= c_{p_i} - \frac{g_i}{T}, &\rho_i &= \exp\left(\frac{m_i}{R} \left( c_{v_i} \ln(T) - s_i \right) \right), \quad &i &= 1, \dots, n.
\end{aligned}
\end{equation}
\end{remark}
\section{Numerical methods\label{sec:Numerical methods}}

\subsection{Finite Volume Discretization}
Equation \(\ref{C4}\) is discretized using the finite volume method. Integrating the governing equation over a control volume \(V_j\) and applying the divergence theorem yields:

\begin{equation}
    \frac{d}{dt} \int_{V_j} \mathbf{u} \, dV 
    + \int_{\partial V_j} \mathbf{F} \cdot \mathbf{n} \, dA 
    - \int_{\partial V_j} \mathbf{G} \cdot \mathbf{n} \, dA 
    = \int_{V_j} \mathbf{S} \, dV,
\end{equation}
where $\mathbf{n}$ is the face normal direction. The volume integrals are approximated as cell-averaged values, and surface integrals are converted to sums over the control volume faces. This results in the semi-discrete formulation:
\begin{equation}
    \frac{d \mathbf{u}_j}{dt} 
    + \frac{1}{V_j} \sum_f \left( \mathbf{F}_f - \mathbf{G}_f \right) A_f 
    = \mathbf{S}_j,
\end{equation}
where \(\mathbf{F}_f\) and \(\mathbf{G}_f\) represent the face-normal fluxes through face \(f\), \(A_f\) is the face area, and \(V_j\) is the control volume. The inviscid flux \(\mathbf{F}_f\) is computed using a numerical flux function, which depends on the left \(\mathbf{u}_L\) and right \(\mathbf{u}_R\) states at the face, $\mathbf{F}_f = \mathbf{F}(\mathbf{u}_L, \mathbf{u}_R)$. The viscous flux \(\mathbf{G}_f\) depends on the gradients of the conserved variables and is typically computed using central difference approximations. 

For time integration, an explicit scheme such as the strong stability-preserving Runge-Kutta (SSP-RK) method \cite{gottlieb2001strong} is applied. The fully discrete update equation is given by:

\begin{equation}
    \mathbf{u}_j^{n+1} = \mathbf{u}_j^n 
    - \frac{\Delta t}{V_j} \sum_f \left( \mathbf{F}_f - \mathbf{G}_f \right) A_f 
    + \Delta t \, \mathbf{S}_j,
\end{equation}
where \(\Delta t\) is the time step size. 

This discretization ensures the conservation of mass, momentum, and energy while accurately resolving fluxes at the interfaces using \(\mathbf{u}_L\) and \(\mathbf{u}_R\). In the next subsection, we will discuss the inviscid fluxes based on entropy-conservation and Entropy-Stable methods.

\subsection{Entropy-Conservative and Entropy-Stable flux}

Writing the Equation \ref{C4} for Euler compressible flow, i.e., viscous effects and source term are not considered:
\begin{equation}\label{es32}
    \frac{\partial \mathbf{\mat {u}}}{\partial t} + \nabla \cdot \mathbf{F}(\mathbf{\mat {u}}) = 0,
\end{equation}
and multiplying on the left by the entropy variables $\bm{v}^{T}$, we obtain
 \begin{align}\label{es33}
\dfrac{\partial U}{\partial t} + \nabla  \cdot \boldsymbol{\mathcal{F}} = 0
\end{align}
Thus, the scheme \ref{es32} is called \textit{entropy conservative} if it satisfies the Equation \ref{es33}, and \textit{entropy stable} if it satisfies the inequality
 \begin{align}\label{es34}
\dfrac{\partial U}{\partial t} + \nabla  \cdot \boldsymbol{\mathcal{F}} < 0 .
\end{align}
In the context of finite volume methods, Tadmor demonstrated how this can be achieved through the judicious choice of numerical flux functions \cite{tadmor1987numerical}. Tadmor demonstrated that if the flux $\mathbf{F}^{ec}$ satisfies the so-called \textit{shuffle} condition:

\begin{equation}\label{L_EC}
[\![ \bm{v}(\mat{u})]\!] \cdot \mathbf{F}^{ec}(\mat{u}_L, \mat{u}_R) = [\![ \bm{\psi}(\mat{u})\cdot \mathbf{n}]\!]
\end{equation}
where $[\![ a ]\!]= a_R -a_L$, then \ref{es33} is satisfied and \ref{es32} is entropy conservative. Also, The finite volume scheme \ref{es32} is entropy stable if the interface flux $\mathbf{F}^{es}$ satisfies an entropy dissipation condition:

\begin{equation}\label{L_ES}
[\![ \bm{v}(\mat{u})]\!] \cdot \mathbf{F}^{es}(\mat{u}_L, \mat{u}_R) < [\![ \bm{\psi}(\mat{u})\cdot \mathbf{n}]\!]
\end{equation}

\subsubsection{Kinetic energy preserving with Entropy-Conservative flux}
An entropy-conserving flux, defined as  
\[
\mathbf{F}^{ec} = [\mathbf{f}_{1,1}^{ec},\ \dots,\ \mathbf{f}_{1,N-1}^{ec},\ \mathbf{f}_{1}^{ec},\  \mathbf{f}_2^{ec},\  \mathbf{f}_3^{ec}],
\]  
which satisfies condition \(\ref{L_EC}\), was introduced in \cite{gouasmi2020formulation,chandrashekar2013kinetic,renac2021entropy}. This flux is given by:  
\begin{equation}\label{ECflux}
\begin{aligned}
\mathbf{f}_{1,i}^{ec}(\mathbf{u}_L, \mathbf{u}_R) &= \rho_i^{\text{ln}}  \ \overline{\mathbf{v}} \cdot \mathbf{n}, \quad 1 \leq i \leq n-1, \\
\mathbf{f}_{1}^{ec}(\mathbf{u}_L, \mathbf{u}_R) &= \sum_{i=1}^n \mathbf{f}_{1,i}^{ec}, \\
\mathbf{f}_{2}^{ec}(\mathbf{u}_L, \mathbf{u}_R) &= \overline{\mathbf{v}} \ \mathbf{f}_{1}^{ec} + p^{ec} \mathbf{n}, \\
\mathbf{f}_{3}^{ec}(\mathbf{u}_L, \mathbf{u}_R) &= \sum_{i=1}^n \left( e_{0i} + \frac{c_{vi}}{\theta^{ln}} + \frac{\overline{\mathbf{v} \cdot \mathbf{v}} }{2} \right) \mathbf{f}_{1,i}^{ec} + p^{ec} \ \overline{\mathbf{v}} \cdot \mathbf{n}.
\end{aligned}  
\end{equation}
Here, the logarithmic mean is denoted as $a^{\text{ln}} = (a_{R} - a_{L})/(\text{ln}(a_{R}) - \text{ln}(a_{L})),$ 
while the arithmetic mean is given by $\overline{a} = (a_{L} + a_{R})/2.$ Additionally, the temperature-related term is defined as \( \theta = 1/T \). A robust numerical method for computing \( a^{\text{ln}} \) was developed by Ismail and Roe \cite{ismail2009affordable}. The interface flux \( \mathbf{F}^{ec} \) ensures entropy conservation for the multi-component system defined in Equation \(\ref{C4}\).  

The Entropy-Conserving (EC) flux introduced in this work extends Chandrasekhar’s EC flux \cite{chandrashekar2013kinetic} to a multi-component formulation for the compressible Euler equations. Chandrasekhar’s original flux also possesses the Kinetic-Energy Preserving (KEP) property, as described by Jameson \cite{jameson2008formulation}, a feature particularly advantageous in turbulence simulations \cite{jameson2008formulation}. For the compressible Euler equations, Jameson \cite{jameson2008formulation} demonstrated that the KEP property is maintained when the momentum flux \( f_{\rho u} \) is expressed as  
\begin{equation}\label{KEP}
f_{\rho u} = p^{ec} + u f_{\rho},
\end{equation}  
where \( f_{\rho} \) represents the mass flux, and \( p^{ec} \) is any consistent average of pressure.  

Extending Jameson’s analysis to a multi-component system follows naturally from \cite{gouasmi2020formulation}. It can be established that the KEP property remains valid when the momentum flux retains the same structural form as in the single-component case, with \( f_{\rho} \) representing the total mass flux \cite{gouasmi2020formulation}. The EC flux derived here satisfies these conditions, with \( p^{ec} \) given by:  
\begin{equation}\label{KEP}
 p^{ec} = \frac{\sum_{i=1}^n {r}_{i} \ {\overline{\rho_i}}}{{\overline{\theta}} }.
\end{equation}

\subsubsection{Entropy-Stable flux}

Tadmor \cite{tadmor2003entropy,fjordholm2012arbitrarily} introduced a significant class of Entropy-Stable schemes, incorporating an entropy variable-based numerical dissipation term to modify the Entropy-Conservative numerical flux. Starting from EC flux $\mathbf{F}^{ec}$, we introduce a general numerical dissipation term to obtain a kinetic energy-preserving and Entropy-Stable flux $\mathbf{F}^{es}$, expressed as
 \begin{align}\label{KEPES1}
\mathbf{F}^{es}(\mathbf{u}_L, \mathbf{u}_R)= \mathbf{F}^{ec}(\mathbf{u}_L, \mathbf{u}_R) 
-\frac{1}{2}\mathbf{D} (\bm{u}_R - \bm{u}_L) ,
\end{align}
where $\mathbf{D}$ represents a suitable dissipation operator. To ensure entropy stability, the dissipation term in \eqref{KEPES1} must be carefully designed so that the numerical flux satisfies the entropy inequality \eqref{es34}. To accomplish this, we redefine the dissipation term to explicitly incorporate the jump in entropy variables:
\begin{align}\label{KEPES2}
\mathbf{D} (\bm{u}_R - \bm{u}_L) \simeq \mathbf{D} \mathbf{A}_0 (\bm{v}_R - \bm{v}_L),
\end{align}
where $\mathbf{A}_0 = {\partial \bm{u}}/{\partial \bm{v}}$ is the symmetric, positive-definite entropy Jacobian matrix that links the conserved and entropy variables. See Section \ref{S22} for more details.

\subsection{Double-Flux Model}
A major challenge in numerical simulations of multi-component flows is the emergence of spurious pressure oscillations, particularly near material interfaces. To identify the root cause and assess the sensitivity of these oscillations, we examine the relationship between internal energy and pressure. This approach was first proposed by Abgrall \cite{billet2003adaptive, abgrall1996prevent} and later extended to a high-order WENO scheme \cite{houim2011low} within the finite difference framework, as well as to a discontinuous Galerkin method \cite{lv2014discontinuous}.

The pressure oscillations arise only for variations in $e_0$ and $\gamma$ occurring in multi-component flows. To mitigate them each cell is treated as perfect gas for the convective timestep. To this end, the equivalent perfect gas properties are calculated for each cell as:

\begin{equation}\label{Estar0}
\begin{aligned}
 c_{pi}^{*} &= \frac{\int_{T_0}^{T} c_{pi}(T) \, dT}{T} , \quad \quad \gamma^{*} = \frac{c_{p}^{*}}{c_{p}^{*} -r}\\
 \end{aligned}
\end{equation}
the total energy per unit volume $E$ can be reformulated as 
\begin{equation}\label{Estar}
\begin{aligned}
E &\overset{\ref{rcpcv}}{=} \sum_{k=1}^{n} Y_i \left(e_{0i} + c_{pi}^{*}T \right) -\frac{p}{\rho} + \frac{|\mathbf{v}|^2}{2} \overset{\ref{Estar0}}{=} e_{0}^{*} + \sum_{k=1}^{n} Y_i \  c_{pi}^{*}\ T  -\frac{p}{\rho} + \frac{|\mathbf{v}|^2}{2}\\
&= e_{0}^{*} + c_{p}^{*} T -\frac{p}{\rho} + \frac{|\mathbf{v}|^2}{2} = e_{0}^{*} + \frac{p}{\rho (\gamma^{*}-1)}+ \frac{|\mathbf{v}|^2}{2}
 \end{aligned}
\end{equation}

Billet and Abgrall \cite{billet2003adaptive} demonstrated that the pressure and velocity of a material interface are preserved if $\gamma^*$ and $\rho e_0^*$ remain constant within each computational cell throughout the entire time step. During reconstruction, these quantities are stored per cell, ensuring consistency in energy calculations. Based on this approach, the flux at the cell interface $j+1/2$ is computed twice---once using the $e_0^*$ and $\gamma^*$ of the left cell and again using those of the right cell. Consequently, the Double-Flux model is inherently non-conservative.

During each stage of the time-marching algorithm the pressure is evaluated from Equation \ref{Estar}. After the timestep the density and pressure are taken from the time-marching scheme. Temperature at the new timestep $T^{n+1}$ is evaluated from the equation of state. Afterwards the caloric equation of state (Equation \ref{Ecal}) is evaluated to compute the internal energy at the new timestep and update $e_0^*$ and $\gamma^*$. 


\begin{remark}Without loss of generality, the formulations presented in \autoref{sec:Governing equations} can be expressed in terms of $\gamma^{*}$ and $e_0^{*}$. For instance, the specific heat at constant volume for a general fluid is given by $c_v^{*} = c_p^{*} - r$.
\end{remark}
\section{Entropy-Stable/Double-Flux scheme\label{sec:Entropy-Stable/Double-Flux scheme}}

The first subsection presents the necessary preliminaries and notation for developing the Entropy-Stable/Double-Flux scheme. The following subsections focus on the formulation of the numerical flux for the multi-component compressible system.

\subsection{Preliminaries and notation}
In this subsection, we introduce key mathematical identities and properties that will be used throughout our analysis. 

\begin{lemma}The jump of the product \(r(\mathbf{Y}) \rho\) across the interface is equal to the product of \(r\) and the jump of \(\rho\) across the interface.
\end{lemma}
 \begin{equation}\label{as44}
[\![r(\mathbf{Y}) \rho]\!] = r [\![ \rho]\!].
\end{equation}
 \begin{proof} To prove \ref{as44}, we use the property of the jump operator $[\![ a b]\!] = \overline{a}[\![ b]\!] + \overline{b}[\![ a]\!]$ and $[\![ r]\!]=0$,  so we can expand $
[\![r(\mathbf{Y}) \rho]\!] = \overline{r}[\![ \rho]\!] + \overline{\rho}[\![ r]\!] = r [\![ \rho]\!].$
 \end{proof}

\begin{lemma}The term $\sum_{i=1}^n r_i  Y_i [\![\text{ln} (Y_{i})]\!]  \rho^{\text{ln}}$ is non-positive. 
\end{lemma}
 \begin{equation}\label{as45}
\sum_{i=1}^n r_i  Y_i [\![\text{ln} (Y_{i})]\!]  \rho^{\text{ln}}  \leq 0
\end{equation}
 \begin{proof} To prove inequality \ref{as45}, we first use the fact that the sum of the jumps of \(Y_i\) across the interface is zero, i.e.,
\begin{equation}\label{sigYi}
\sum_{i=1}^n [\![Y_{i}]\!] = \sum_{i=1}^n  Y_{i,R} - \sum_{i=1}^n Y_{i,L} = 1 - 1 = 0.
\end{equation} 
Thus, we can proceed by
\begin{equation}\label{Profas45}
\sum_{i=1}^n r_i  Y_i [\![\text{ln} (Y_{i})]\!]  \rho^{\text{ln}}  \leq \sum_{i=1}^n  [\![\text{ln} (Y_{i})]\!]  r  \rho^{\text{ln}} =  \sum_{i=1}^n 
\frac{[\![Y_{i}]\!]}{Y_i^{\text{ln}}} r  \rho^{\text{ln}} \leq \sum_{i=1}^n 
[\![Y_{i}]\!] \frac{r \rho^{\text{ln}}}{{Y_i^{min}}^{\text{ln}}}\overset{\ref{sigYi}}{=} 0
\end{equation}
where $Y_i ^{min} = \text{min}_{1 \leq Y_i \leq n} Y_i$.
 \end{proof}

\begin{lemma}The jump of the product of the inverse temperature and the Gibbs function of species i, is given by:
 \begin{equation}\label{as46}
[\![\theta g_{i}]\!] = e_{0i}^* [\![\theta]\!] + c_{vi}^* [\![\text{ln} \theta]\!] +r_i [\![\text{ln} \rho_i]\!] 
\end{equation}
\end{lemma}
 \begin{proof}To prove equation \ref{as46}, we have
\begin{equation}
 \begin{aligned}
[\![\theta g_{i}]\!] &\overset{\ref{gib}}{=} [\![\theta (e_i^* + r_i T - T s_i)]\!] = [\![\theta  e_{0i}^* + \ c_{pi}^* - s_i)]\!]\\
&\overset{\ref{s}}{=} [\![\theta  e_{0i}^* + \ c_{pi}^* + c_{vi}^* \text{ln}\theta + r_i \text{ln} \rho_i]\!] =  e_{0i}^* [\![\theta]\!] + c_{vi}^* [\![\text{ln} \theta]\!] +r_i [\![\text{ln} \rho_i]\!]  
\end{aligned}
\end{equation}
\end{proof}
\begin{lemma}With consider $\beta = \rho/p = \theta / r$, the following relation holds:

 \begin{equation}\label{as47}
\sum_{i=1}^{n}\left(e_{0i}^* - \frac{c_{vi}^*}{\theta^{\text{ln}}}\right) Y_i \ \rho^{\text{ln}} =\left( e_{0}^{*} + \frac{1}{(\gamma^{*}-1) \beta^{\text{ln}} }\right) \rho^{\text{ln}}
\end{equation}
\end{lemma}
 \begin{proof}With the observation that $\theta^{\text{ln}} = r \ \beta^{\text{ln}}$, we get
\begin{equation}
  \begin{aligned}
\sum_{i=1}^{n}\left(e_{0i}^{*} - \frac{c_{vi}^{*}}{\theta^{\text{ln}}}\right) Y_i \ \rho^{\text{ln}} &\overset{\ref{rcpcv}}{=} \left(e_{0}^* - \frac{c_{v}^*}{\theta^{\text{ln}}}\right) \ \rho^{\text{ln}} = \left(e_{0}^* - \frac{c_{v}^*}{r \ \beta^{\text{ln}}}\right) \ \rho^{\text{ln}}\\
&=\left( e_{0}^{*} + \frac{1}{(\gamma^{*}-1) \beta^{\text{ln}} }\right) \rho^{\text{ln}}
\end{aligned}
\end{equation}
\end{proof}

\subsection{Combine the Entropy-Stable and Double-Flux formulations}
In the following, we develop an inviscid numerical flux for the spatial discretization of Equation \ref{C4}. Our goal is to propose a new type of inviscid numerical flux that is Entropy-Stable and benefits from the advantages of the Double-Flux model. To this end, we first consider the numerical flux $\mathbf{F}^{ec}$ and modify it based on mixture variables. The Entropy-Stable flux for the Double-Flux model  \( \mathbf{F}^{es^*} = [\mathbf{f}_{1,1}^{es^*},\ \dots,\ \mathbf{f}_{1,n-1}^{es^*},\ \mathbf{f}_{1}^{es^*},\  \mathbf{f}_2^{es^*},\  \mathbf{f}_3^{es^*}] \) is given by  

\begin{equation}\label{ESDLfluxtot}
\begin{aligned}
\mathbf{f}_{1,i}^{es^*}(\mathbf{u}_L, \mathbf{u}_R)&= Y_i \ \rho^{\text{ln}}  \ \overline{\mathbf{v}} \cdot \mathbf{n}, \quad 1 \leq i \leq n-1, \\
 \mathbf{f}_{1 }^{es^*}(\mathbf{u}_L, \mathbf{u}_R)& =\sum_{i=1}^n  \mathbf{f}_{1,i}^{es^*}= \sum_{i=1}^n Y_i \ \rho^{\text{ln}}  \ \overline{\mathbf{v}} \cdot \mathbf{n} = \rho^{\text{ln}}  \ \overline{\mathbf{v}} \cdot \mathbf{n}\\
 \mathbf{f}_{2 }^{es^*}(\mathbf{u}_L, \mathbf{u}_R)&= \overline{\mathbf{v}} \ \mathbf{f}_{1 }^{es^*} + \overline{p} \ \mathbf{n}, \\
\mathbf{f}_{3 }^{es^*}(\mathbf{u}_L, \mathbf{u}_R) &= \left( e_{0}^{*} + \frac{1}{(\gamma^{*}-1) \beta^{\text{ln}} } +\frac{\overline{\mathbf{v} \cdot \mathbf{v}}}{2}\right) \mathbf{f}_{1 }^{es^*} +\overline{p}\ \overline{\mathbf{v}}\cdot \mathbf{n} 
\end{aligned}
\end{equation}
where $\overline{p} = {\overline{\rho}}/{\overline{\beta}} =
{\overline{p\theta}}/{\overline{\theta}} $.

\begin{theorem}The finite volume discretization of the  multi-component compressible Euler equations with inviscid numerical flux {\ref{ESDLfluxtot}} is Entropy-Stable, meaning that the total generalized entropy is non-increasing over time.
\end{theorem}

 \begin{proof}The flux $\mathbf{F}^{es^*}$ is symmetry and consistency follow from symmetry and consistency of the logarithmic mean and average operator. To prove that \ref{ESDLfluxtot} is Entropy-Stable \ref{L_ES} by following:
 \begin{equation}\label{po2}
[\![ \bm{\psi}(\mat{u})\cdot \mathbf{n}]\!] \overset{\ref{po1}}{=} [\![ r \rho \mathbf{v}\cdot \mathbf{n}]\!] \overset{\ref{rcpcv}}{=} 
(\overline{p\theta}[\![ \mathbf{v}]\!] + [\![ r \rho]\!] \overline{\mathbf{v}}) \cdot \mathbf{n} \overset{\ref{as44}}{=} 
\overline{p\theta}[\![ \mathbf{v}]\!] \cdot \mathbf{n} + r [\![ \rho]\!] \overline{\mathbf{v}} \cdot \mathbf{n} \overset{\ref{rcpcv}}{=} 
\overline{p\theta}[\![ \mathbf{v}]\!] \cdot \mathbf{n} + \sum_{i=1}^n Y_i r_i [\![ \rho]\!] \overline{\mathbf{v}} \cdot \mathbf{n}
\end{equation}
Using short notations for the flux components and the definition of $\mathbf{F}^{es^*}$ in \ref{ESDLfluxtot} we get 
\begin{equation}\label{ESDLfluxtot2}
\begin{aligned}
    [\![ \bm{v}(\mat{u})]\!] \cdot \mathbf{F}^{es^*}(\mat{u}_L, \mat{u}_R) &= 
    \sum_{i=1}^{n - 1} [\![ \theta (g_i - g_n)]\!] \mathbf{f}_{1,i}^{es^*} +[\![\theta g_n - \frac{\mathbf{v} \cdot \mathbf{v}}{2} \theta]\!] \mathbf{f}_{1 }^{es^*}
     + [\![ \theta \mathbf{v} ]\!] \cdot \mathbf{f}_{2 }^{es^*}- [\![ \theta ]\!] \mathbf{f}_{3 }^{es^*}\\
    & \overset{\ref{as46}}{=} \sum_{i=1}^{n - 1}\left((e_{0i}^*-e_{0n}^*)[\![ \theta ]\!] + (c_{vi}^* -c_{vn}^*)  [\![\text{ln} \theta ]\!] + r_i [\![\text{ln} \rho_i]\!]  - r_{n} [\![\text{ln} \rho_n]\!] \right)\mathbf{f}_{1,i}^{es^*} \\
    &+ \left( e_{0n}^* [\![\theta]\!] + c_{vn}^* [\![\text{ln} \theta]\!] +r_n [\![\text{ln} \rho_n]\!]  \right) \mathbf{f}_{1 }^{es^*} - \left(\frac{\overline{\mathbf{v} \cdot \mathbf{v}}}{2}[\![ \theta ]\!] + \overline{\theta} \overline{\mathbf{v}}\cdot[\![\mathbf{v} ]\!] \right) \mathbf{f}_{1 }^{es^*}\\
    &+\left(\overline{\theta}[\![\mathbf{v} ]\!] + \overline{\mathbf{v}}[\![\theta ]\!]\right) \cdot \mathbf{f}_{2 }^{es^*}- [\![ \theta ]\!] \mathbf{f}_{3 }^{es^*}\\
     & \overset{\ref{ESDLfluxtot}}{=} \sum_{i=1}^{n}\left( e_{0i}^* [\![\theta]\!] + c_{vi}^* [\![\text{ln} \theta]\!] +r_i [\![\text{ln} \rho_i]\!]  \right) \mathbf{f}_{1,i}^{es^*}\\
     &- \left(\frac{\overline{\mathbf{v} \cdot \mathbf{v}}}{2}[\![ \theta ]\!] + \overline{\theta} \overline{\mathbf{v}}\cdot[\![\mathbf{v} ]\!] \right) \mathbf{f}_{1 }^{es^*}
    +\left(\overline{\theta}[\![\mathbf{v} ]\!] + \overline{\mathbf{v}}[\![\theta ]\!]\right) \cdot \mathbf{f}_{2 }^{es^*}- [\![ \theta ]\!] \mathbf{f}_{3 }^{es^*}.
    \end{aligned}
\end{equation}
Using Equations  \ref{ESDLfluxtot2} and \ref{po2} together with the observation that $[\![\text{ln} \rho_i]\!] = [\![\text{ln} Y_i]\!] + [\![\text{ln} \rho]\!]$,  $[\![ \bm{v}(\mat{u})]\!] \cdot \mathbf{F}^{es^*}(\mat{u}_L, \mat{u}_R) - [\![ \bm{\psi}(\mat{u})\cdot \mathbf{n}]\!] $ becomes
\begin{equation}\label{ESDLfluxtot3}
\begin{aligned}
   &\sum_{i=1}^{n}  r_i[\![\text{ln} \rho]\!] \left(\mathbf{f}_{1,i}^{es^*} - Y_i \ \rho^{\text{ln}} \ \overline{\mathbf{v}} \cdot \mathbf{n} \right)  + \sum_{i=1}^{n}  r_i[\![\text{ln} Y_i]\!] \ \mathbf{f}_{1,i}^{es^*}
   + \overline{\theta}[\![\mathbf{v} ]\!] \cdot \left( \mathbf{f}_{2 }^{es^*} - \overline{\mathbf{v}} \cdot \mathbf{f}_{1 }^{es^*} - \frac{\overline{p\theta}}{\overline{\theta}} \mathbf{n}  \right)\\
   &+ [\![\text{ln} \theta]\!] \left(\sum_{i=1}^{n}\left(e_{0i}^* \ \theta^{\text{ln}}- c_{vi}^* \right)\mathbf{f}_{1,i}^{es^*} - \theta^{\text{ln}} 
   \left(\mathbf{f}_{3 }^{es^*} - \overline{\mathbf{v}} \cdot \mathbf{f}_{2}^{es} + \frac{\overline{\mathbf{v} \cdot \mathbf{v}}}{2}\mathbf{f}_{1 }^{es^*}\right) \right)
   \end{aligned}
\end{equation}
inserting $\mathbf{F}^{es^*}$ from \ref{ESDLfluxtot} into \ref{ESDLfluxtot3}, and using
also Equations \ref{as47}, \ref{as45}, so we obtain 
\begin{equation}\label{ESDLfluxtot4}
    [\![ \bm{v}(\mat{u})]\!] \cdot \mathbf{F}^{es^*}(\mat{u}_L, \mat{u}_R) - [\![ \bm{\psi}(\mat{u})\cdot \mathbf{n}]\!] = \sum_{i=1}^n r_i  Y_i [\![\text{ln} (Y_{i})]\!]  \rho^{\text{ln}} \  \leq 0
    \end{equation}
which ends the proof and the numerical flux $\mathbf{F}^{es^*}$ is Entropy-Stable. Note that for $n=1$, the numerical flux in entropy-consercative.
 \end{proof}
\begin{remark}Equation \ref{ESDLfluxtot} computes the numerical flux more efficiently than Equation \ref{ECflux}, requiring only two logarithmic mean evaluations instead of $n+1$. This makes it advantageous for $n>2$, as logarithmic mean calculations are computationally expensive. 
\end{remark}
\begin{remark}It is evident that \(\mathbf{F}^{es^*}\) satisfies the condition \ref{KEP}. Thus, the numerical flux in Equation \ref{ESDLfluxtot} is designed to preserve kinetic energy \cite{chandrashekar2013kinetic}.
\end{remark}

\subsection{Derivation of dissipation at interfaces for Entropy-Stable/Double-Flux method}

We here follow the Equation {\ref{KEPES1}} and look for a dissipation of the form {\ref{KEPES2}}. Different choices of the dissipation matrix $\mathbf{D}$ yield different numerical fluxes. Two commonly used approaches are the Lax-Friedrichs (LF) \cite{fleischmann2020low,badrkhani2023matrix} and Roe schemes \cite{roe1981approximate}. It is well established that the LF scheme introduces higher numerical dissipation than the Roe scheme. A hybrid dissipation formulation blends these schemes adaptively, applying the more dissipative LF term near strong shocks to enhance robustness while leveraging the less dissipative Roe term in smooth regions and near rarefaction waves or contact discontinuities to improve accuracy. This hybrid dissipation approach is defined as
 \begin{align}\label{KEPES3}
\mathbf{D} = \overline{\bm{R}} \: \vert  \overline{\bm{ \Lambda}} \vert  \:\overline{\bm{R}}^{-1},
\end{align}
where the diagonal matrix of eigenvalues is given by
\begin{equation}
\vert  \overline{\bm{ \Lambda}} \vert   =	(1-\theta)\,\vert  {\overline{ \bm{ \lambda}}} \vert  + \theta \,\vert  {{ \lambda}_{max}} \vert  \bm{I}, \qquad
\theta = \sqrt{\left| \frac{[\![p ]\!]}{2 \overline{p} } \right|}
\end{equation}
and the parameter $\theta \in [0,1]$ is defined using a simple local pressure indicator. Here, $\overline{\bm{R}}$ is the matrix of right eigenvectors, and $\overline{  \bm{ \lambda}}$ represents the diagonal matrix of eigenvalues of the inviscid flux Jacobian, which is $
{\overline{  \bm{ \lambda}} =  \text{diag}\left(\overline{\text{v}}_\text{n}, \cdots,\overline{\text{v}}_\text{n},\overline{\text{v}}_\text{n} -\overline{c},\overline{\text{v}}_\text{n}, \overline{\text{v}}_\text{n},\overline{\text{v}}_\text{n}, \overline{\text{v}}_\text{n} +\overline{c} \right)} $ and $\lambda_{max}$ denotes the largest eigenvalue of $\overline{  \bm{ \lambda}}$. Since only the mass flux and energy flux eigenvalues need to be equal to ensure kinetic energy dissipation \cite{chandrashekar2013kinetic}, we chose eigenvalues of the kinetic energy preserving with Entropy-Stable flux is form 
\begin{equation}
{\overline{  \bm{ \lambda}}} ={ \bm{ \lambda}^{KEPES}} = \text{diag}\left(\overline{\text{v}}_\text{n}, \cdots,\overline{\text{v}}_\text{n},\overline{\text{v}}_\text{n} +\overline{c},\overline{\text{v}}_\text{n}, \overline{\text{v}}_\text{n},\overline{\text{v}}_\text{n}, \overline{\text{v}}_\text{n} +\overline{c} \right). 
\end{equation}

The average components of the right eigenvectors are given by

\begin{equation}
\overline{\bm{R}}=
\begin{bmatrix}
1       & \cdots    & 0     & Y_1      & Y_1       & 0         & 0         & Y_1       \\
\vdots  & \ddots    &\vdots & \vdots   & \vdots    & \vdots    & \vdots    & \vdots  \\
0       &\cdots     &1      &Y_{n-1}  & Y_{n-1}    & 0         & 0           & Y_{n-1} \\
0 & \cdots      &0          &1 & 1 & 0 & 0 &   1 \\
0 & \cdots      &0          &\overline{\text{v}}_1 - \overline{c} \ \text{n}_x  & \overline{\text{v}}_1    & - \text{n}_y & - \text{n}_z & \overline{\text{v}}_1 + \overline{c} \ n_x \\
0 & \cdots      &0          &\overline{\text{v}}_2 - \overline{c} \ \text{n}_y & \overline{\text{v}}_2    & \ \  \text{n}_x & 0 & \overline{\text{v}}_2 + \overline{c} \ \text{n}_y \\
0 & \cdots      &0          &\overline{\text{v}}_3 - \overline{c} \ \text{n}_z & \overline{\text{v}}_3   &   0 &  \ \ \text{n}_x &  \overline{\text{v}}_3 + \overline{c} \ \text{n}_z \\
z_1 & \cdots    &z_{n-1}    &\overline{h} - \overline{c} \ \overline{\text{v}}_\text{n} & \frac{1}{2} \overline{{\Vert \text{v}_{tot} \Vert}^2} 
& \overline{\text{v}}_2 \text{n}_x -  \overline{\text{v}}_1 \text{n}_y
& \overline{\text{v}}_3 \text{n}_x -  \overline{\text{v}}_1 \text{n}_z
& \overline{h} \ + \overline{c} \ \overline{\text{v}}_\text{n}
\end{bmatrix}
\end{equation}
where
\begin{equation}
\overline{c} =	\sqrt \frac{\gamma^* \overline{p}}{\rho ^{\ln}}, \qquad
\overline{h} = e_0^* + \frac{\gamma^*}{ \beta ^{\ln} (\gamma^* - 1)} +  \frac{1}{2}  \overline{{\Vert \text{v}_{tot} \Vert}^2},\qquad
\overline{\text{v}}_\text{n}= \overline{\text{v}}_1 \text{n}_x + \overline{\text{v}}_2 \text{n}_y + \overline{\text{v}}_3 \text{n}_z 
\end{equation}
and $z_i =- \rho ^{\ln}	 \frac{\partial p}{\partial \rho Y_i}/\frac{\partial p}{\partial e}$. The pressure derivatives $\frac{\partial p}{\partial \rho Y_i}$ and $\frac{\partial p}{\partial e}$ are obtained by fixing other variables in the transformed EoS $p=p(\rho, e, \rho Y_1 , \dots,\rho Y_{n-1})$. 

To construct an Entropy-Stable numerical flux with matrix-based dissipation, a relationship between the entropy Jacobian $\mathbf{A}_0$ and the right eigenvectors $\overline{\bm{R}}$ is required. According to Barth's eigenvector scaling theorem \cite{barth1999numerical}, there exists a positive diagonal scaling matrix such that
\begin{equation}\label{KEPES4}
\mathbf{A}_0 = \overline{\bm{R}} \:\overline{\bm{T}} \ \overline{\bm{R}}^{T}
\end{equation}
where $\overline{\bm{T}}$ is defined based on reference \cite{gouasmi2020formulation}. However, the diagonal scaling matrix is defined based on the parameter \(\gamma^*\), whereas in reference \cite{gouasmi2020formulation}, it was originally defined using \(\gamma\). This modification allows for greater flexibility in adapting the scaling properties to doubel-flux model while maintaining consistency with the previous approach.

Combining \eqref{KEPES2}, \eqref{KEPES3}, and \eqref{KEPES4}, we express the hybrid dissipation term as
\begin{equation}\label{KEPES5}
\mathbf{D} \mathbf{A}_0 (\bm{v}_R - \bm{v}_L) = \overline{\bm{R}} \: \vert  \overline{\bm{ \Lambda}} \vert  \:\overline{\bm{R}}^{-1} ( \overline{\bm{R}} \:\overline{\bm{T}} \ \overline{\bm{R}}^{T})(\bm{v}_R - \bm{v}_L) = \overline{\bm{R}} \: \vert  \overline{\bm{ \Lambda}} \vert  \:\overline{\bm{T}} \ \overline{\bm{R}}^{T} (\bm{v}_R - \bm{v}_L).
\end{equation}
Thus, the Entropy-Stable/Double-Flux flux formulation is given by
 \begin{align}\label{KEPES6}
\mathbf{F}^{es/df}(\mathbf{u}_L, \mathbf{u}_R)= \mathbf{F}^{es^*}(\mathbf{u}_L, \mathbf{u}_R) 
-\frac{1}{2}
  \overline{\bm{R}} \: \vert  \overline{\bm{ \Lambda}} \vert  \:\overline{\bm{T}} \ \overline{\bm{R}}^{T} (\bm{v}_R - \bm{v}_L) .
\end{align}
This formulation guarantees entropy stability by enforcing the entropy inequality in a discrete form. To demonstrate this, we contract the semi-discrete approximation of the multi-component compressible equations with the entropy variables \eqref{Esv}. The entropy stability of the flux \eqref{KEPES6} results in the discrete entropy inequality.
\begin{equation}\label{KEPES7}
 \begin{aligned}
\dfrac{\partial U}{\partial t} + \nabla  \cdot \boldsymbol{\mathcal{F}} &\leq  [\![ \bm{v}(\mat{u})]\!] \cdot \mathbf{F}^{es^*}(\mat{u}_L, \mat{u}_R) - [\![ \bm{\psi}(\mat{u})\cdot \mathbf{n}]\!]\\
&-\frac{1}{2} (\bm{v}_R - \bm{v}_L)^{T} \:
  \overline{\bm{R}} \: \vert  \overline{\bm{ \Lambda}} \vert  \:\overline{\bm{T}} \ \overline{\bm{R}}^{T} (\bm{v}_R - \bm{v}_L) \leq 0
  \end{aligned}
 \end{equation} 
Since the last term right-hand side of \eqref{KEPES7} is a quadratic form scaled by a negative factor, the entropy inequality is discretely satisfied.

\section{Implementation aspects \label{sec:Implementation aspects}}

This section outlines the implementation of the Entropy-Stable/Double-Flux solver in OpenFOAM, integrating Cantera for chemical kinetics, thermodynamics, transport properties, and Adaptive Mesh Refinement (AMR) for efficiency. The solver employs a density-based approach for accurate shock capturing. AMR enhances computational performance by dynamically refining the mesh based on error criteria. The section also details the high-performance computing setup used for simulations.

\subsection{Solver structure}
The developed Entropy-Stable/Double-Flux solver is represented by a coupling library implemented in OpenFOAM \cite{weller_1998} which connects with the chemical program Cantera \cite{goodwin2002cantera} and local mesh refinement (AMR) \cite{heylmun2023}. Density-based solvers are crucial for accurately simulating compressible flows, particularly in scenarios involving strong shocks like those found in turbomachinery. Unlike pressure-based methods, density-based algorithms excel in shock-capturing capabilities, making them the preferred choice for transonic and hypersonic flow calculations. In our work, we adapted the density-based solver from \cite{heylmun_blastfoamguide_2021} to develop our scheme. Figure~\ref{fig:flowchart} shows the schematic structure of the Entropy-Stable/Double-Flux solver.
\begin{figure}
    \centering
    {{\includegraphics[width=1.0\textwidth]{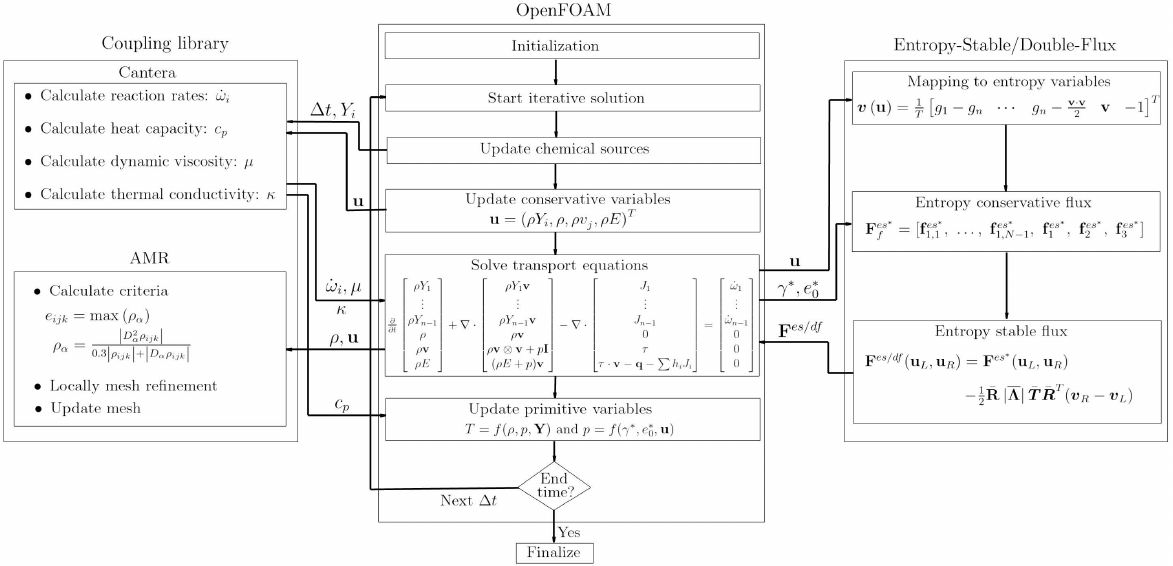} }}%
    \caption{Simplified flowchart illustrating the coupling of the density-based solver in OpenFOAM with Adaptive Mesh Refinement (AMR), the Cantera library, and the Entropy-Stable/Double-Flux solver.}%
    \label{fig:flowchart}%
\end{figure}

\subsection{Adaptive mesh refinement functionality}

With  Adaptive mesh refinement (AMR), the local spatial resolution can be dynamically controlled. This allows the maximization of the computational eﬃciency of the overall simulation as higher resolution is placed only where it is needed. An example grid structure with colored density gradients for the shock-wave and reacting helium gas cylinder interaction problem is shown in Figure\ref{fig:AMR_mesh}.

For multi-component equation flexibility, AMR’s single criterion refinement was extended to multiple different criteria. The criteria include uniquely selected suitable fields and  their gradients  features such as boxes or domain boundaries as well as the maximum and minimal refinement levels on each individual criterion. The refinement approach proposed by Sun and Takayama \cite{sun1999conservative} is employed to evaluate the maximum error. The error is quantified using a specific metric computed for each cell in a three-dimensional grid, defined as 
\begin{flalign}
  e_{ijk} = \max \left( \rho_{\alpha} \right), \quad \forall \alpha \in \{x, y, z, xy, yx, xz, zx, yz, zy\}, 
\end{flalign}
where the error components \(\rho_{\alpha}\) are computed using a general finite-difference stencil:  
\begin{equation}
    \rho_{\alpha} = \frac{\left| D_{\alpha}^2 \rho_{ijk} \right|}{0.3 \left| \rho_{ijk} \right| + \left| D_{\alpha} \rho_{ijk} \right|},
\end{equation}
with \(D_{\alpha}^2\) and \(D_{\alpha}\) representing the second- and first-order finite-difference approximations along the respective directions. The second-order difference is computed as  $D_x^2 \rho_{ijk} = \rho_{i-1,j,k} - 2 \rho_{i,j,k} + \rho_{i+1,j,k}$, with analogous expressions for other directions. The first-order difference is given by  
$ D_x \rho_{ijk} = \rho_{i+1,j,k} - \rho_{i-1,j,k}$. A block is marked for refinement when the highest error within the block,  
$
    e_{\max} = \max(e_{ijk}),
$
exceeds the refinement threshold \(e_{\text{ref}}=0.1\).

\begin{figure}
    \centering
    \includegraphics[width=0.80\textwidth]{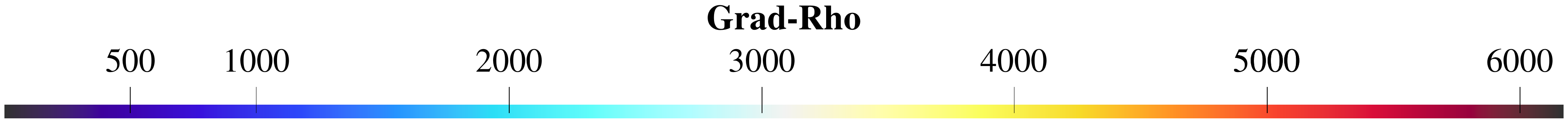}%
    \vspace{0.2cm} 
    {{\includegraphics[width=0.98\textwidth]{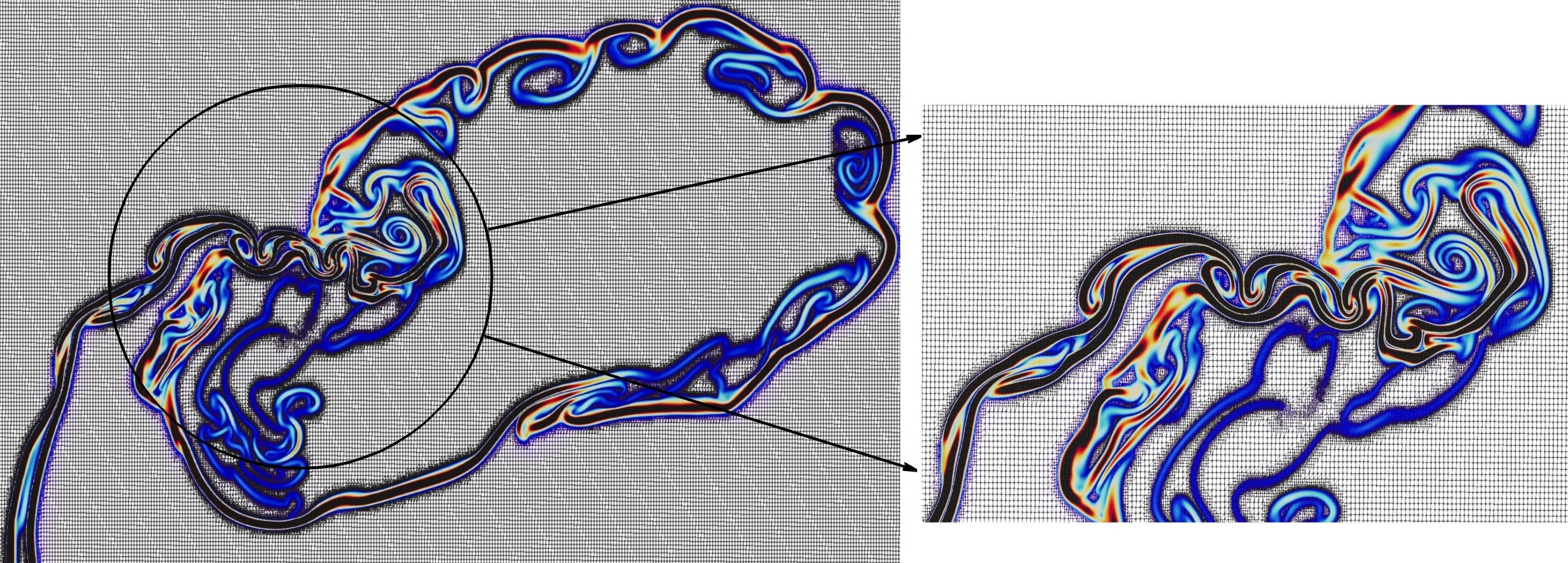} }}%
    \caption{A grid structure with colored  density gradients for the shock-wave and reacting helium bubble  interaction problem at $t=700 \: \mu \text{s}$ with three levels of refinement.}%
    \label{fig:AMR_mesh}%
\end{figure}

\section{Numerical results \label{sec:Numerical results}}

In this section, we demonstrate the effectiveness of our Entropy-Stable/Double-Flux approach through a series of test cases with increasing complexity. The goal is to validate the scheme's accuracy and stability across progressively more challenging compressible flow problems. We start with single-component 1D cases \ref{res1} to establish a baseline, then introduce multi-component 1D problems \ref{res2} to test the scheme’s handling of species interactions. We then extend to multi-component 2D shock-dominated flows \ref{res3}, and finally to a multi-component 3D underexpanded jet \ref{res4}, representing a complex, realistic flow scenario\footnote{Computational experiments were performed on the Lichtenberg II (Phase 1) cluster at TU Darmstadt using multiple nodes, each with dual Intel Xeon Platinum 9242 CPUs (48 cores, 2.3 GHz) and up to 384 GB RAM. The code was compiled with GCC 10.2.0, HWLOC 2.7.1, OpenMPI 4.0.2, and Cantera 3.1.0}.

This progression--from 1D to 3D and from single- to multi-component--allows for a clear evaluation of the scheme’s performance under growing numerical and physical challenges.

We use a third-order accurate strong stability preserving Runge-Kutta time integration scheme (SSP-RK3).To ensure numerical stability in a computational run, the time step $\Delta t$ is determined using the standard finite volume $CFL$ condition:
\begin{equation}
    \Delta t = CFL \times \frac{h}{\lambda^{\max}} 
\end{equation}
where $h$ is the characteristic cell length and $\lambda^{\max}$ represents the maximum wave speed at time step $n$ \cite{badrkhani2025matrix}. The $CFL$ number is a user-defined coefficient that regulates the stability of the simulation. In this study, all computations are performed with $CFL = 0.75$.

\subsection{Discontinuous profile on a periodic domain} \label{res1}
We examine the discrete evolution of entropy  in a single-component test case by evolving a discontinuous initial profile to final time $t= 2s$ on the domain [-1,1]. The initial conditions for density and velocity in non-dimensional form are prescribed as follows:
\begin{equation}
\rho\left( x,t\right) = \left \{\begin{array}{ll}
      3 \quad \vert x \vert < 0.5\\
      2 \quad \text{otherwise}
    \end{array}
  \right.,\qquad \mathbf{v}\left( x,t\right) =0,\qquad p\left(x,t \right) = \rho^{\gamma}.
\end{equation}
Periodic boundary conditions are enforced in order to examine the evolution of entropy over longer time periods. The spatial resolution has been fixed to 500. We examine the change in entropy over time. Numerical experiments in \citep{gassner2016well} suggest that the discrete change in entropy over time should converge to zero as the timestep decreases for Entropy-Conservative flux. We compute the convergence rate of $\Delta{s}(t) = s(x,t) - s(x,0) $ to zero with respect to the timestep $ \Delta{t}$, as shown in Figure \ref{fig:Error_S}.

For this test case, we consider the numerical flux $\mathbf{F}^{es/df} = \mathbf{F}^{es^*}$, as defined in Equation \ref{ESDLfluxtot}. To illustrate this, we analyze the flux for a single component ($n=1$) by setting $Y_1 = 1$ and utilizing Equation \ref{ESDLfluxtot4}. This allows us to express
\begin{equation}
    [\![ \bm{v}(\mat{u})]\!] \cdot \mathbf{F}^{es^*}(\mat{u}_L, \mat{u}_R) - [\![ \bm{\psi}(\mat{u})\cdot \mathbf{n}]\!] = \sum_{i=1}^n r_i Y_i [\![\ln(Y_i)]\!]  \rho^{\ln} = 0.
\end{equation}
This result confirms that, for this test case, the numerical flux $\mathbf{F}^{es^*}$ is Entropy-Conservative. Furthermore, Figure \ref{fig:Error_S} provides numerical evidence supporting this conclusion.

\begin{figure}
    \centering
    {{\includegraphics[width=0.6\textwidth]{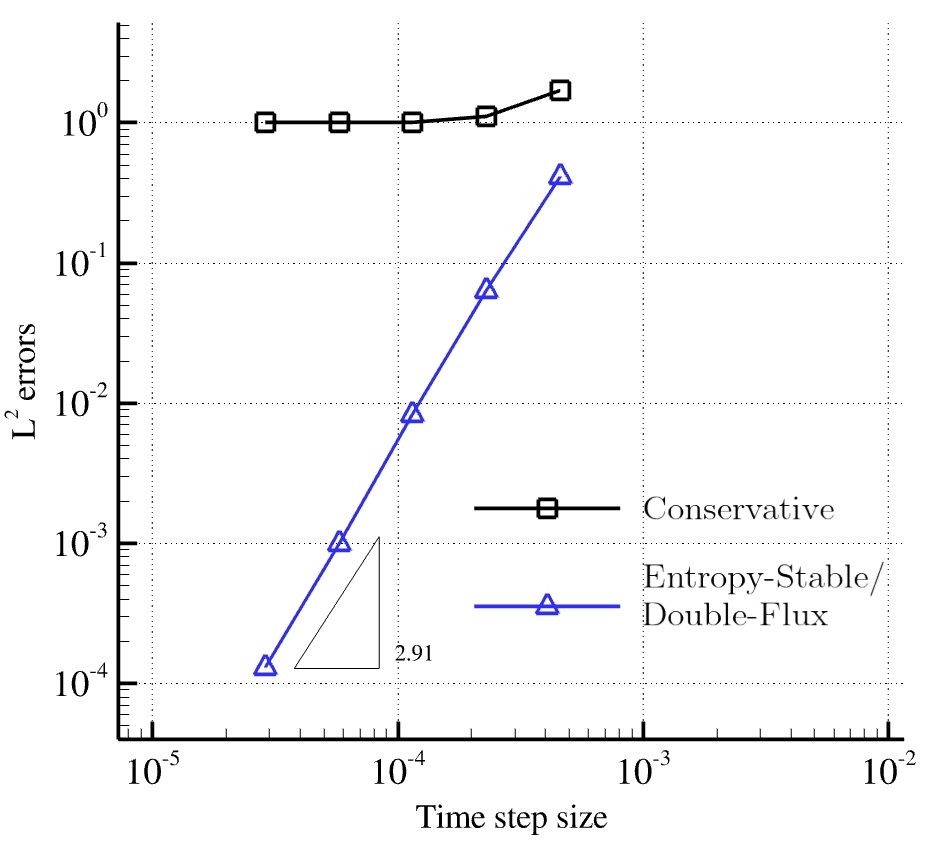} }}%
    \caption{Change in entropy $s(t)$ for both Entropy-Stable/Double-Flux and conservative method using a Lax-Friedrichs (LF) fluxes. The convergence of the change in entropy $s(t)$ at the final time $t = 2s$ converges to zero as $O\left( \Delta{t}^{2.91}\right) $ which is the 3rd order time-stepper used.}%
    \label{fig:Error_S}%
\end{figure}

\subsection{Moving interface} \label{res2}

The next test problem is the advection of a contact discontinuity (constant velocity and constant pressure) separating two different species. This example illustrates the classic oscillation phenomena when encountering a discontinuous $\gamma$ with the Entropy-Stable scheme, as well as how this is overcome by the Entropy-Stable/Double-Flux method. Since the Double-Flux method is already known to prevent oscillations in this test case \cite{houim2011low, lv2014discontinuous}, we focus on comparing the Entropy-Stable and Entropy-Stable/Double-Flux methods. This highlights the specific effect of the Double-Flux correction when applied within the Entropy-Stable framework. The initial conditions are given by:

\begin{equation}
\begin{array}{ll}
Y_i\left( x,t=0\,\mathrm{s}\right) = \left \{\begin{array}{ll}
      Y_{H_{2}}=1  \quad 0\,\mathrm{m} < x  < 0.05\,\mathrm{m}\\
      Y_{N_{2}}=1 \quad \text{otherwise}
    \end{array}
  \right.,\quad \mathbf{v}\left( x,t=0\,\mathrm{s}\right) =100\,\mathrm{m/s},\\\\ 
  p\left(x,t=0\,\mathrm{s} \right) = 1\mathrm{atm}, \quad T\left(x,t=0\,\mathrm{s} \right) = 300\mathrm{K}.

\end{array}
\end{equation}
The entire domain ranges $[-0.05,0.5] \,\mathrm{m}$ with 1000 cells and a periodic boundary condition is employed. The pressure and velocity profiles at $t = 0.001 s$ are shown in Figure \ref{fig:Moving interface}. 

The observed overshoots and undershoots in the profiles are characteristic of conservative schemes \cite{abgrall2001computations} as well as Entropy-Stable methods \cite{jameson2008formulation}. These anomalies are inherent to multi-component flows and should be expected. However, the Double-Flux model uniquely maintains accurate pressure and velocity values \cite{wang2025adaptive}, free from any fluctuations, though this comes at the expense of introducing conservation errors \cite{houim2011low}. Notably, the combination of the Entropy-Stable method and the Double-Flux model is shown to accurately capture material interfaces without the occurrence of artificial oscillations. This stability is achieved by freezing the variables $\gamma^*$ and $e_0^*$ throughout the entire timestep, in conjunction with the Double-Flux algorithm within the Entropy-Stable framework.

\begin{figure}
    \centering
    \subfloat[\centering Pressure, $p$]{{\includegraphics[width=0.5\textwidth]{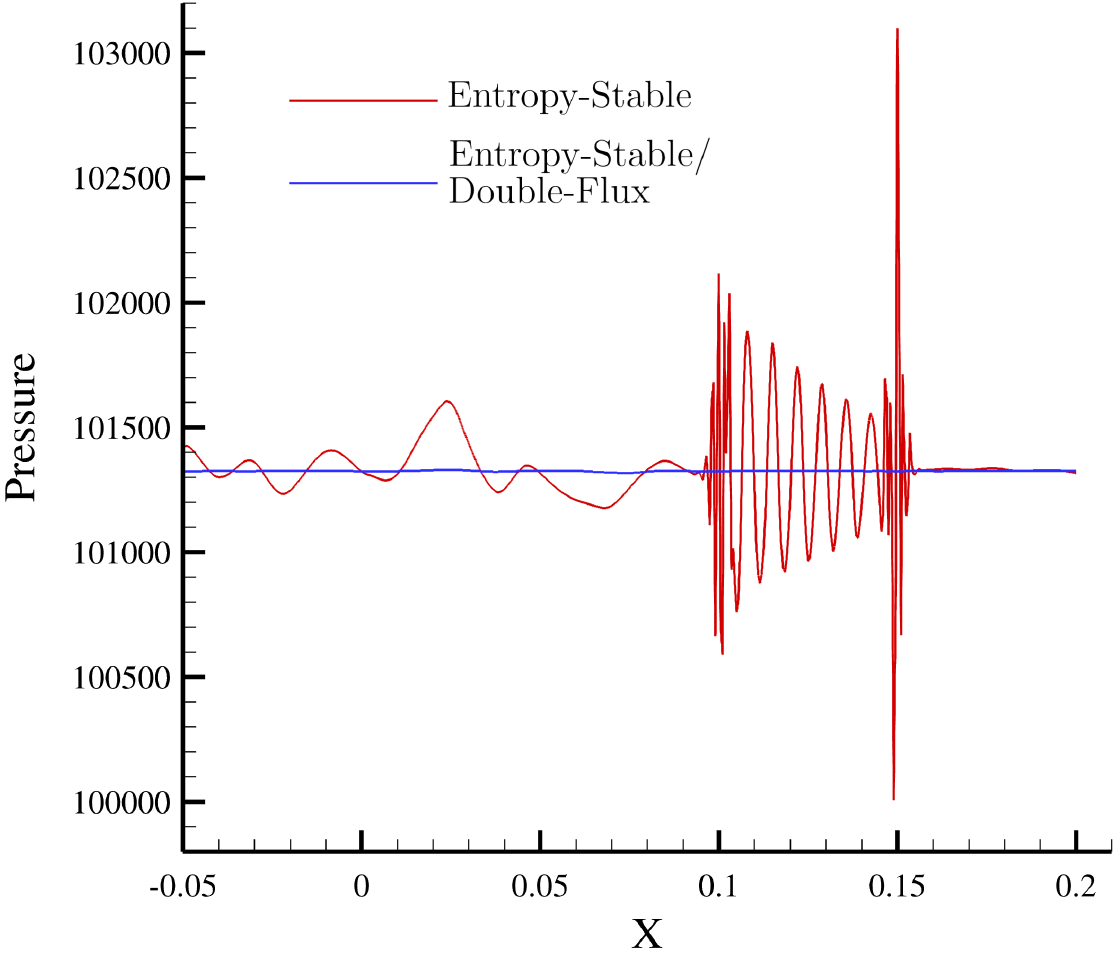} }}%
        \subfloat[\centering Velocity, $\mathbf{v}$]{{\includegraphics[width=0.5\textwidth]{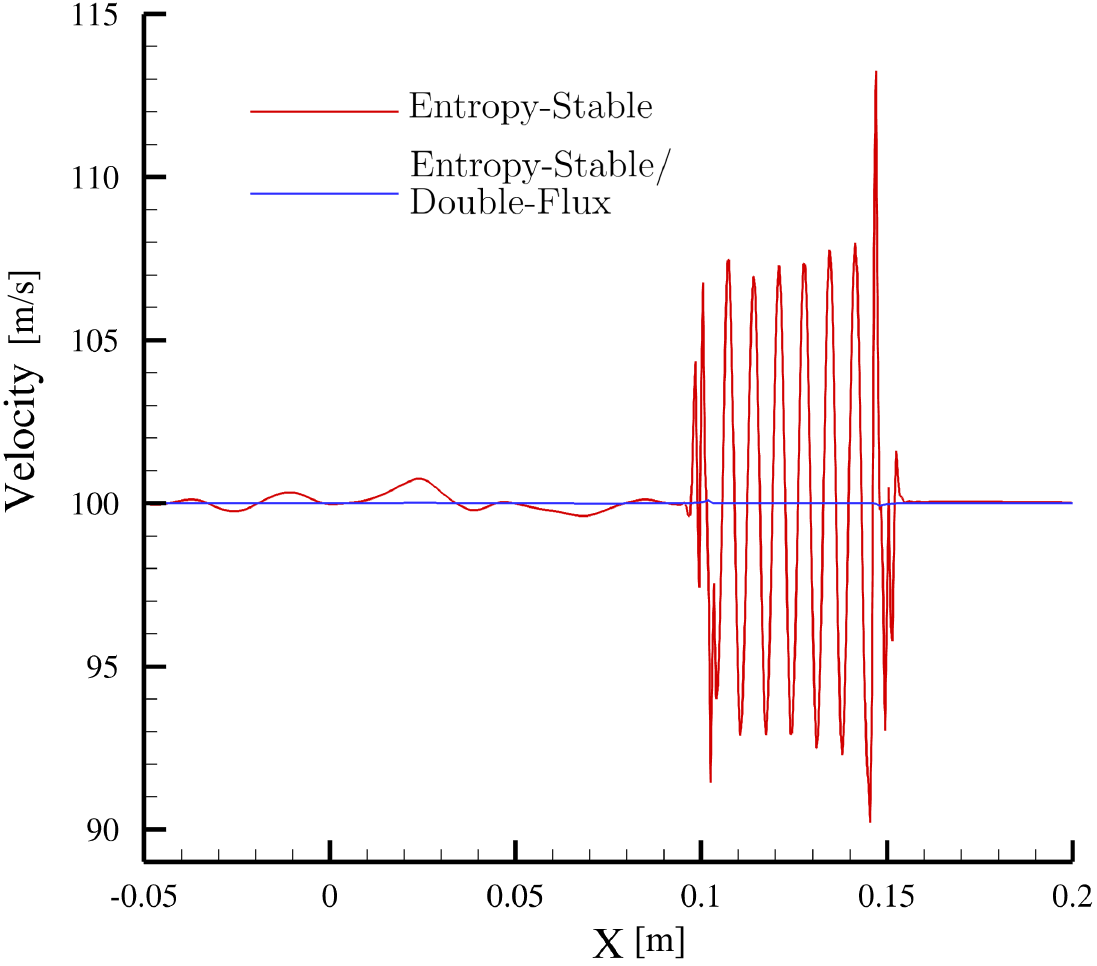} }}%
    \caption{Pressure and velocity profiles for the one-dimensional moving interface problem, $t = 0.001 s$.}%
    \label{fig:Moving interface}%
\end{figure}

\subsection{2D Shock-bubble interaction} \label{res3}

This study examines the two-dimensional interaction between a Mach 1.22 shock wave and a helium bubble suspended in air. The experimental configuration, originally investigated by Haas and Sturtevant \cite{haas1987interaction}, has been widely used for validating numerical schemes in compressible multiphase flow simulations \cite{johnsen2006implementation,marquina2003flux,quirk1996dynamics}. The computational domain is defined as $\Omega = (0,0.325)\: \text{m} \times (0,0.0455) \: \text{m}$. Initially, the shock is positioned at $x = 0.225 \: \text{m}$, while the helium bubble is centered at $x_0 = 0.175 \: \text{m}$. The domain boundaries impose slip-wall and symmetric conditions at the top and centerline, while supersonic inflow and outflow conditions are enforced at the right and left boundaries, respectively. The schematic bubble position and shock-wave location are shown in Figure~\ref{fig:HE_Testcase}. The initial conditions are summarized in Table~\ref{tab:initial_conditions}.

\begin{table}[h]
\captionsetup{justification=raggedright, singlelinecheck=false}
    \caption{Initial conditions for the two-dimensional shock-bubble interaction problem.}
    \label{tab:initial_conditions}
\begin{tabular*}{\textwidth}{l@{\extracolsep{\fill}}lccc}
        \toprule
        Quantity & Pre-shock air & Post-shock air & Helium bubble \\
        \midrule
        $\rho \: \:$ [kg/m$^3$]   & 1.29  & 1.7756        & 0.2347 \\
        $\text{v}_1$ [m/s]           & 0     & $(M_2 \: c_2 - M_1 \: c_1)$       & 0 \\
        $\text{v}_2$ [m/s]           & 0     & 0             & 0 \\
        $p \:\:$ [bar]           & 1     & $p_2/p_1$      & 1  \\
        $T\:$ [K]             & 300   & 300 $T_2/T_1$           & 300 \\
        $Y_{N_2}$        & 0.215   & 0.215           & 0.000 \\
        $Y_{O_2}$           & 0.785   & 0.785           & 0.000 \\
        $Y_{He}$            & 0.000   & 0.000           & 1.000 \\        \bottomrule
    \end{tabular*}
\end{table}
where $c = \sqrt{\gamma R T}$ and the normal shock ratios, $T_2/T_1$ and $p_2/p_1$, and post shock Mach number, $M_2$, can be calculated from the isentropic flow relations for air $\gamma= 1.4$ which yields, for a $M_1 = 1.22$ normal shock, $T_2/T_1 = 1.14054$, $p_2/p_1 =1.56979$, and $M_2 = 0.829986$.

Adaptive mesh refinement was used with 1-3 levels of refinement. The computational mesh consists of quadratic elements with three refinement levels $h =(100, 50, 25) \: \mu \text{m}$. Simulations employ SSP-RK3 to integrate the parabolic terms, with a temporal resolution of $\Delta t = 5 \times 10^{-9}$ s, corresponding to a maximum $CFL$ number of $0.75$ for the finest mesh. 
\begin{figure}
    \centering
    {{\includegraphics[width=0.99\textwidth]{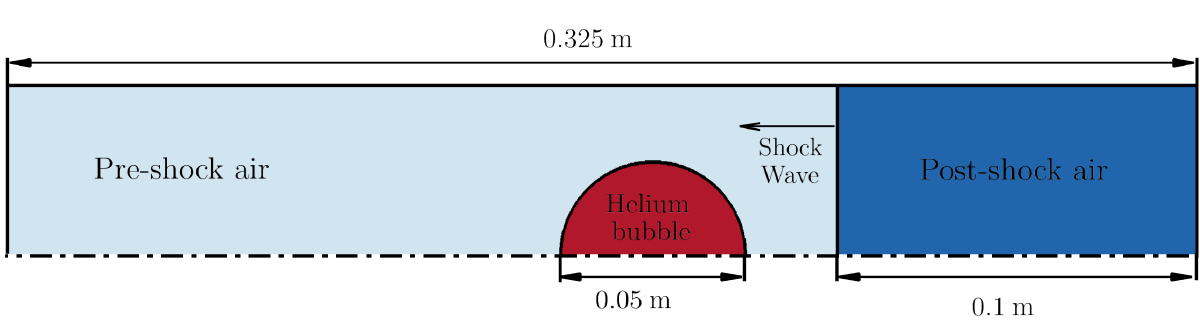} }}%
    \caption{Schematic of the two-dimensional Shock-bubble interaction problem.}%
    \label{fig:HE_Testcase}%
\end{figure}
\subsubsection{Simulation results and analysis}

Figure~\ref{fig:gradrho_HE} shows the helium mass fraction images and the density gradients at different times after the shock reached the bubble. At $t = 350 \: \mu \text{s}$, three key points in the solution are identified as the \textbf{downstream}, \textbf{jet}, and \textbf{upstream} locations, as shown in Figure~\ref{fig:gradrho_HE}. The \textbf{downstream point} corresponds to the leftmost position of the helium bubble, while the \textbf{jet point} marks its rightmost position along $y = 0$ m. The \textbf{upstream point} represents the bubble’s farthest right location within the entire domain. When the shock interacts with the helium bubble, these points shift over time, eventually leading to the merging of the jet and downstream locations. The trajectories of these points, depicted in Figure \ref{fig:Jet_Up_Down} for the Entropy-Stable/Double-Flux scheme with $h=25 \: \mu \text{m}$ solution, closely match the solid reference lines from \cite{terashima2009front}.

\begin{figure}
    \centering
    \includegraphics[width=0.40\textwidth]{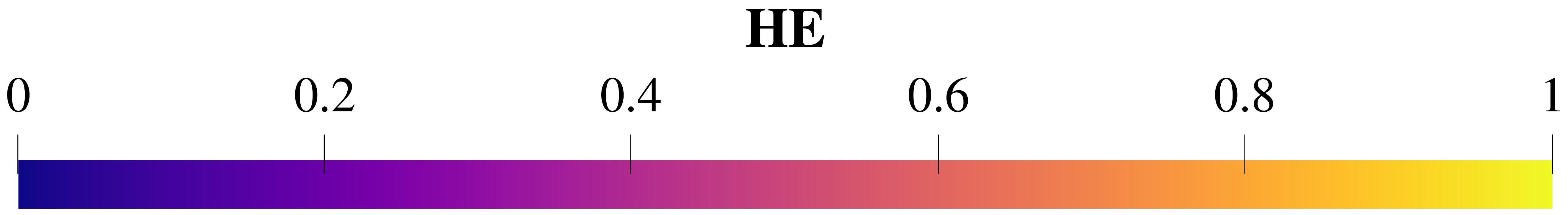} \hspace{1.52cm}
    \includegraphics[width=0.40\textwidth]{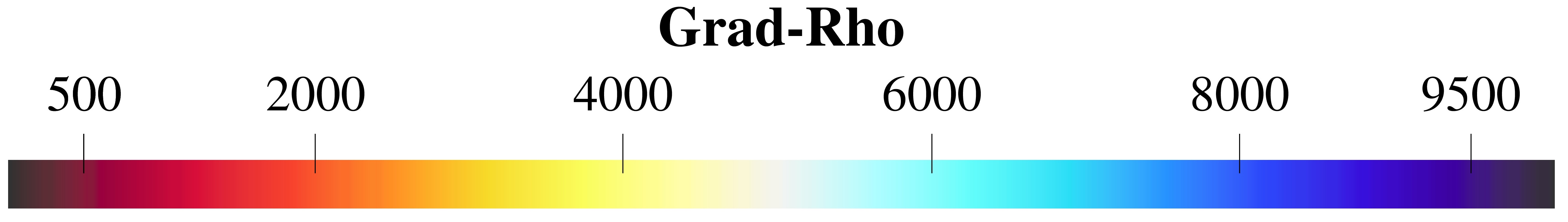}%
    \vspace{0.2cm} 
    \includegraphics[width=0.495\textwidth]{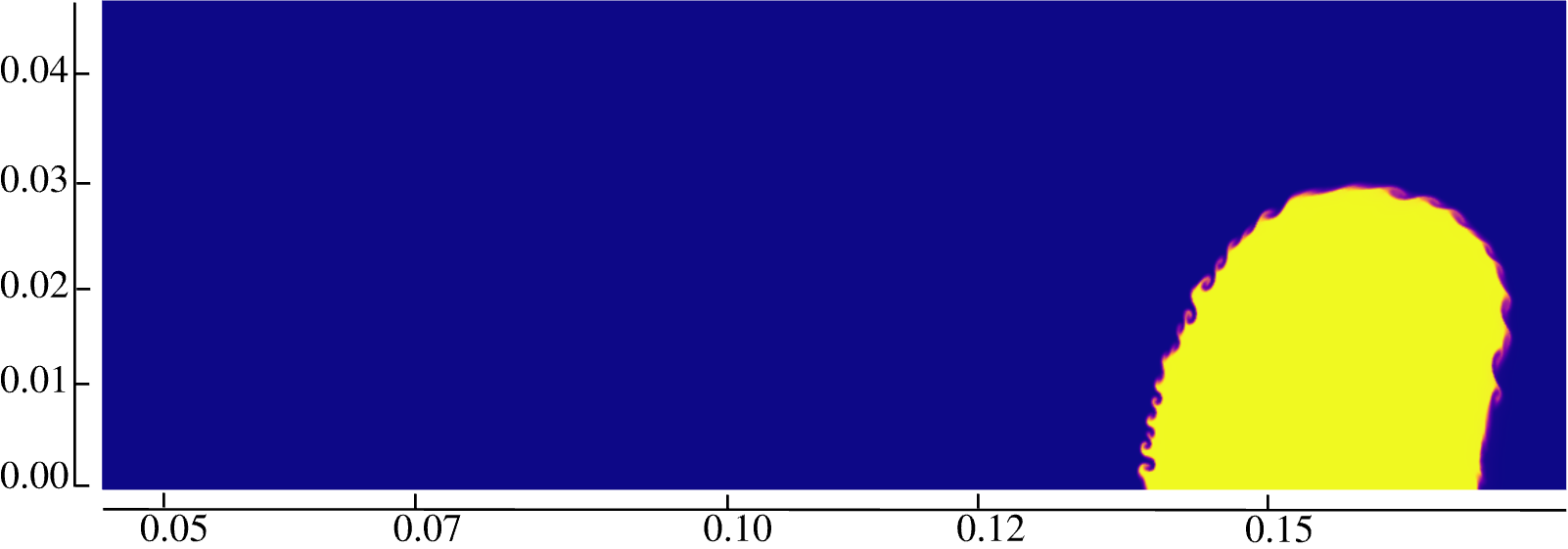}
    \includegraphics[width=0.495\textwidth]{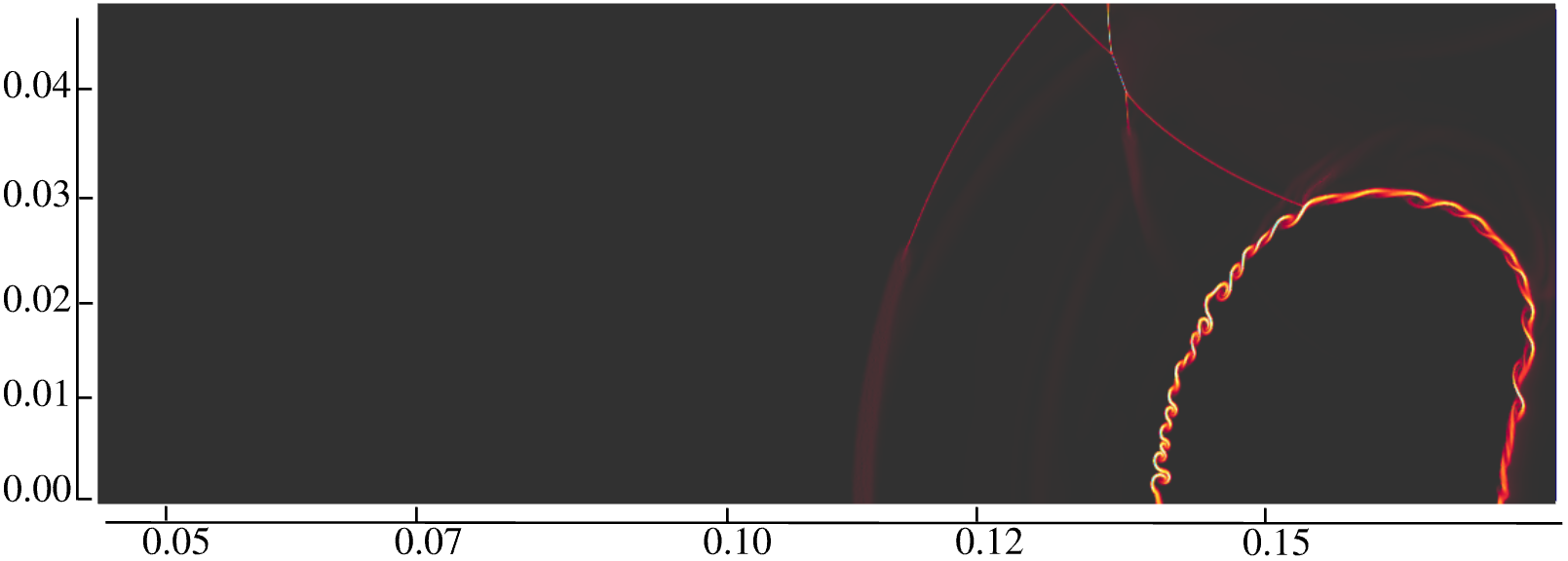}%
    \vspace{0.2cm} 
    \includegraphics[width=0.495\textwidth]{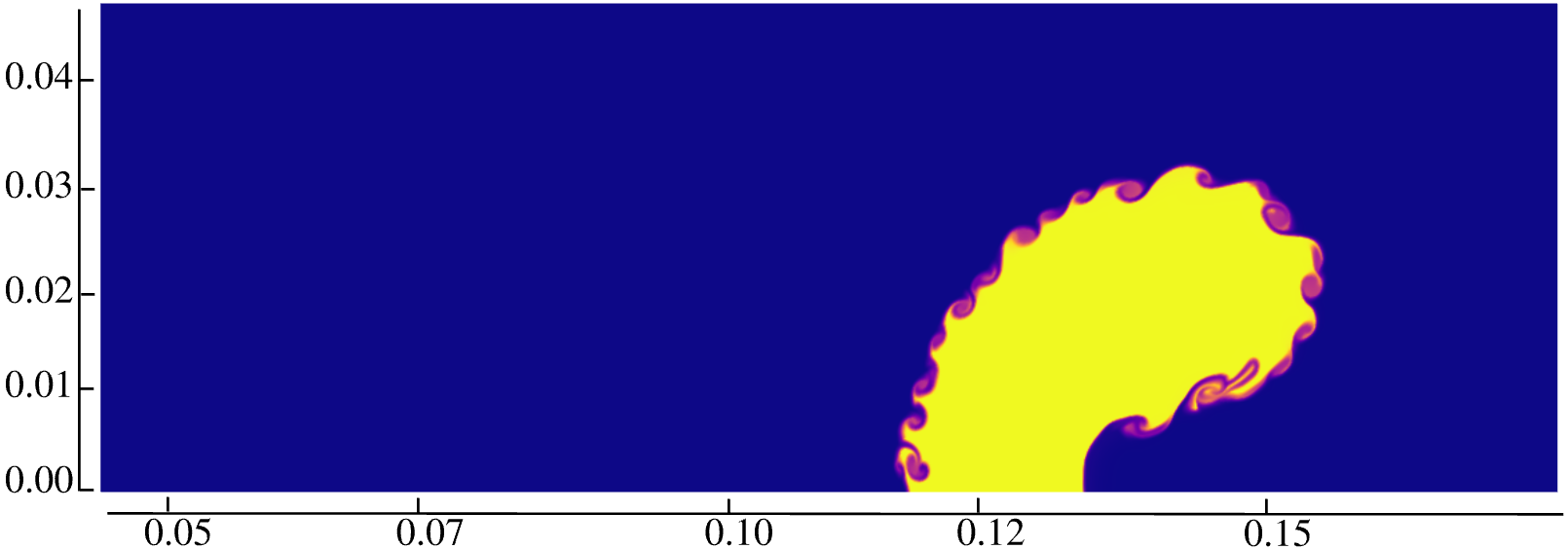}
    \includegraphics[width=0.495\textwidth]{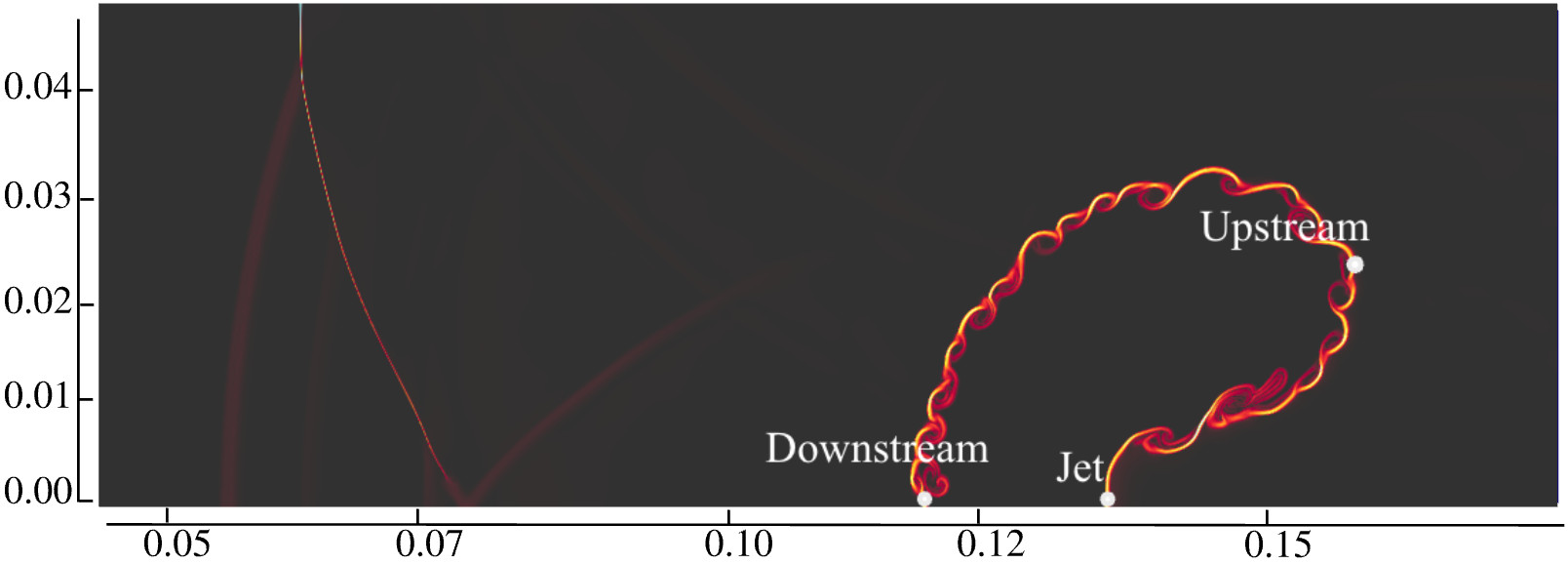}%
    \vspace{0.2cm} 
    \includegraphics[width=0.495\textwidth]{ 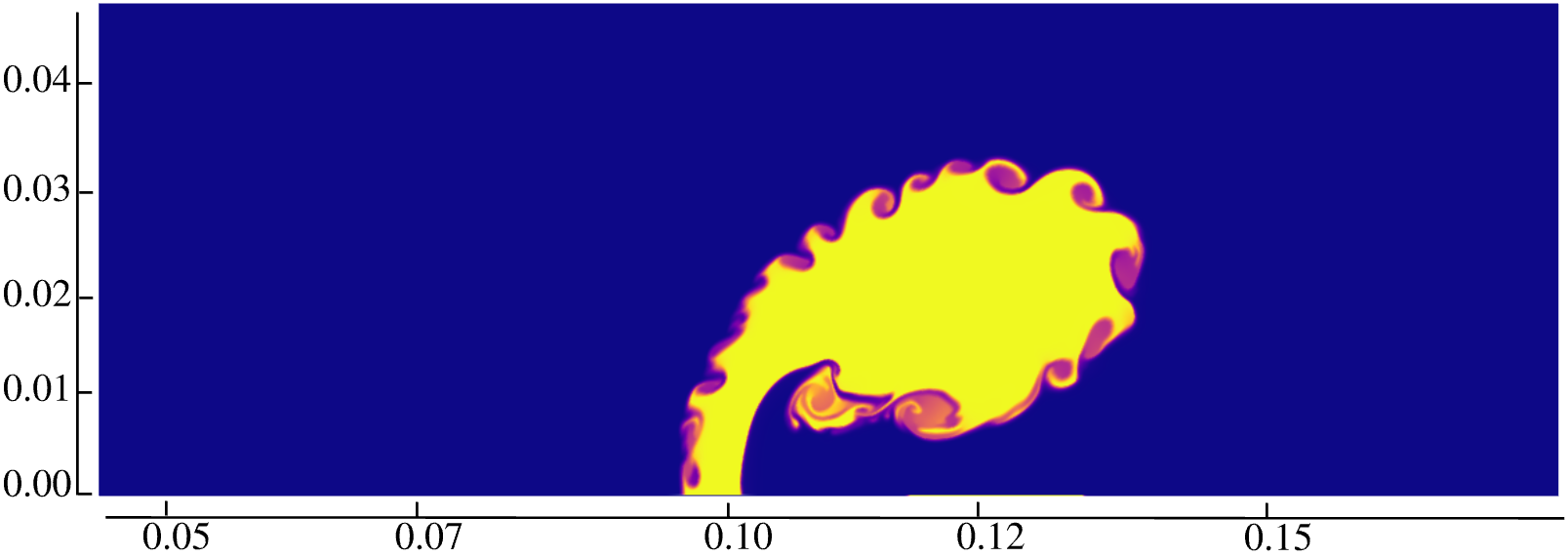}
    \includegraphics[width=0.495\textwidth]{ 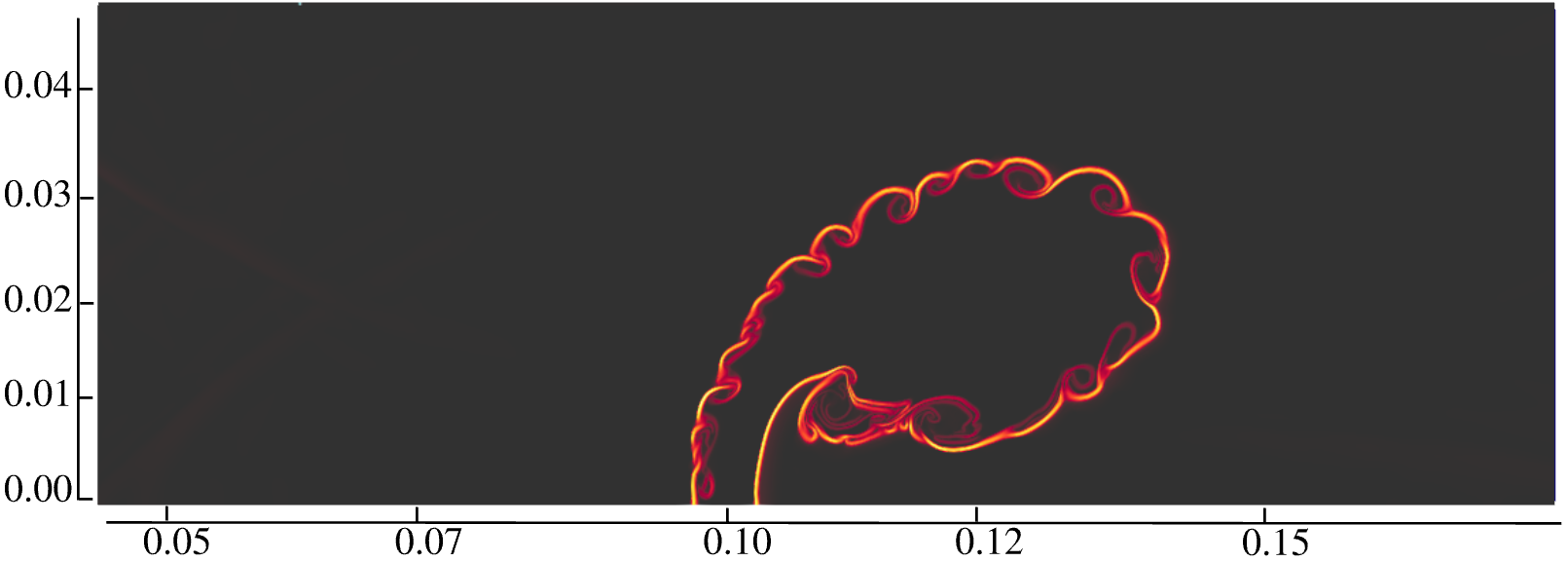}%
    \vspace{0.2cm} 
    \includegraphics[width=0.495\textwidth]{ 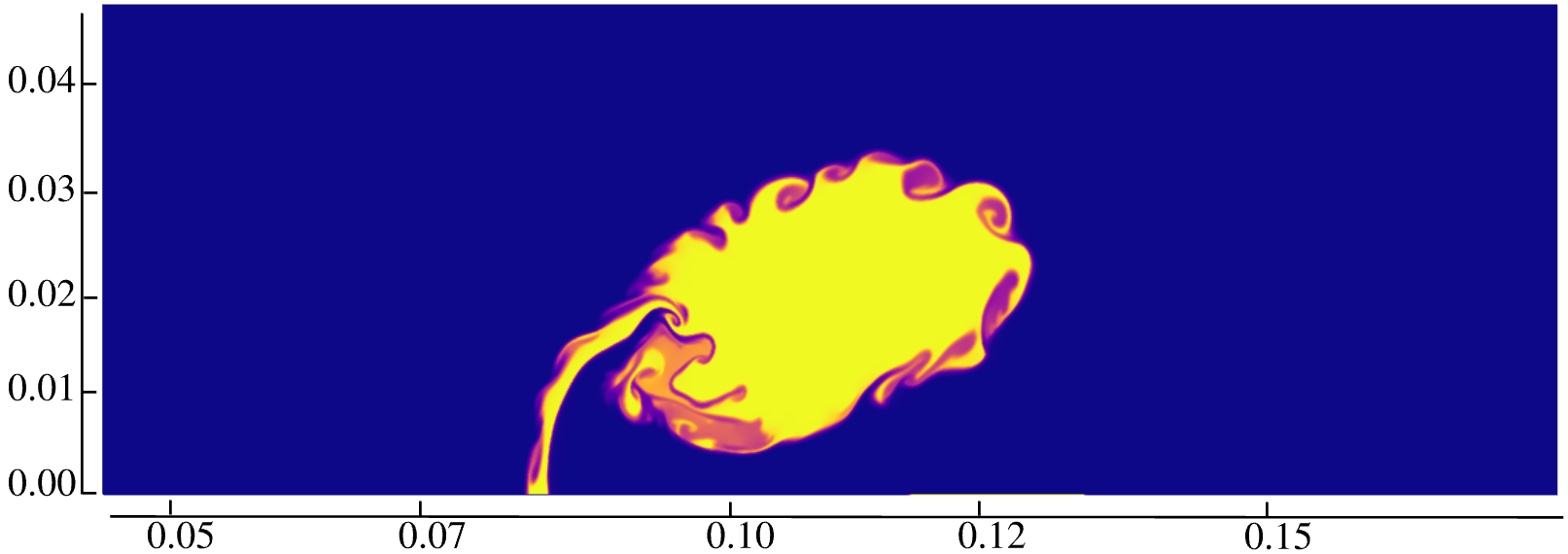}
    \includegraphics[width=0.495\textwidth]{ 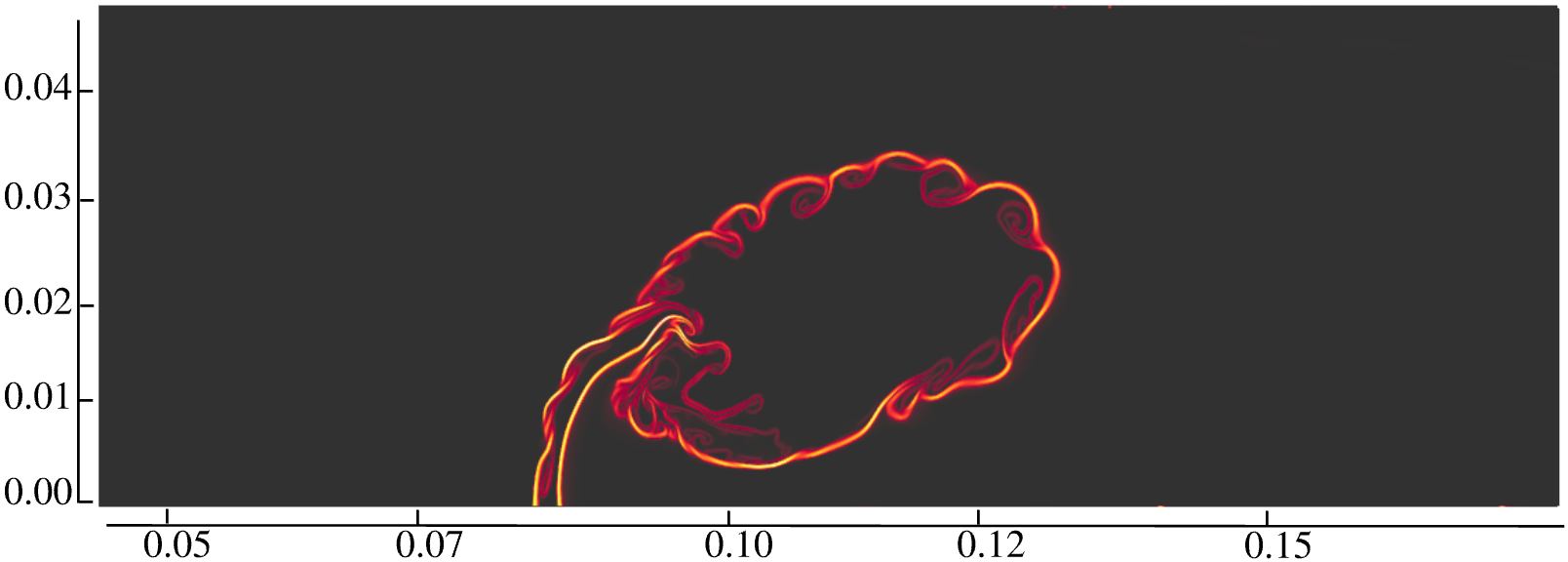}%
    \vspace{0.2cm} 
    \includegraphics[width=0.495\textwidth]{ 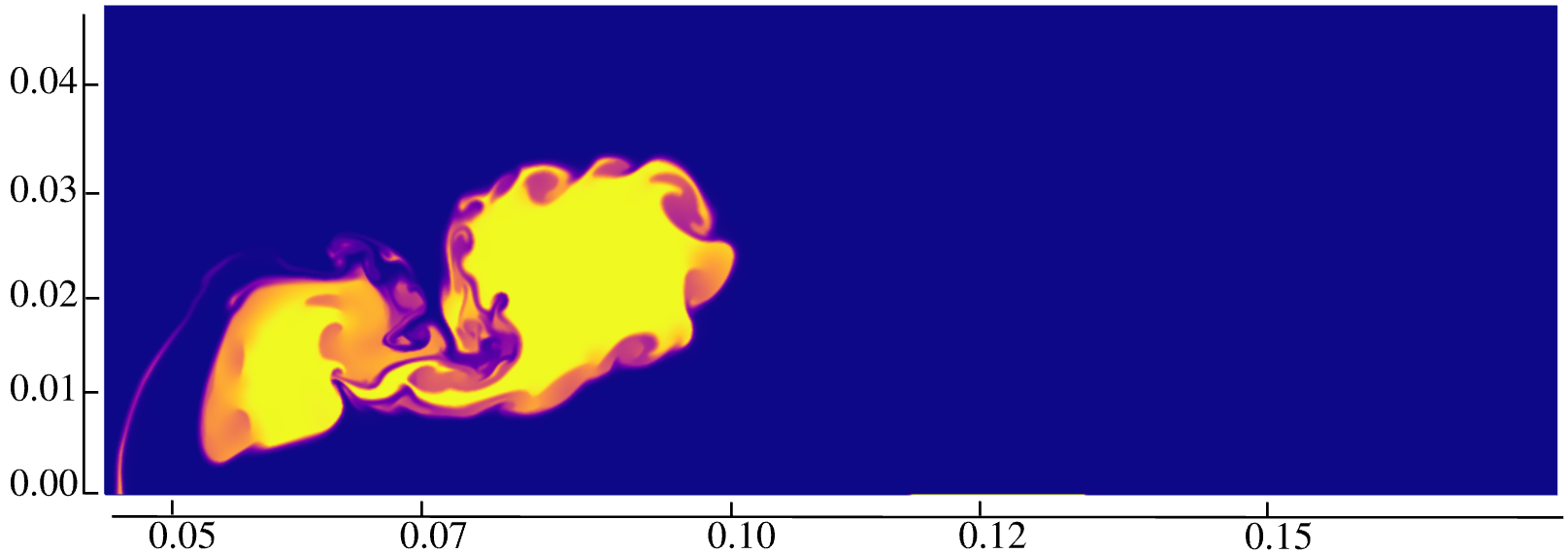}
    \includegraphics[width=0.495\textwidth]{ 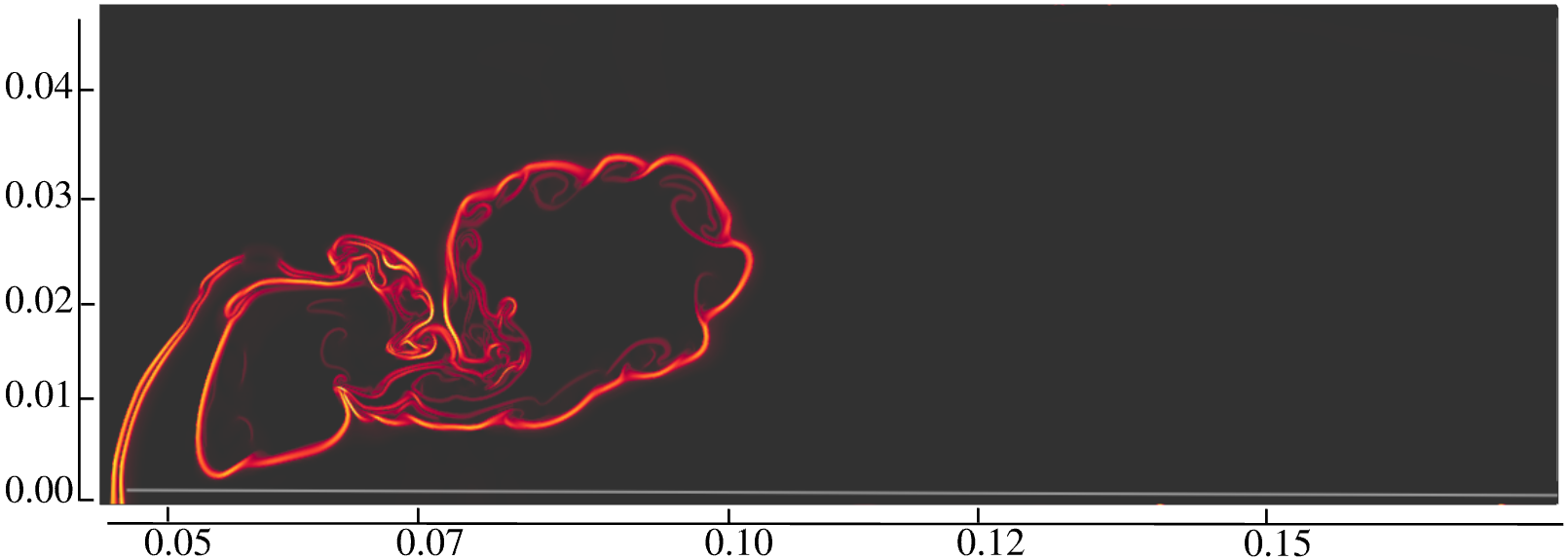} \caption{Evolution of helium mass fraction (left) and density gradient magnitude (right) for the interaction with the shock of the helium bubble in air test case computed with the Entropy-Stable/Double-Flux scheme. Time stamps from top to bottom: $t = 200 \: \mu \text{s}$, $t = 350 \: \mu \text{s}$, $t = 500 \: \mu \text{s}$, $t = 600 \: \mu \text{s}$ and $t = 850 \: \mu \text{s}$.}
    \label{fig:gradrho_HE}
\end{figure}

\begin{figure}
    \centering
    {{\includegraphics[width=0.49\textwidth]{ 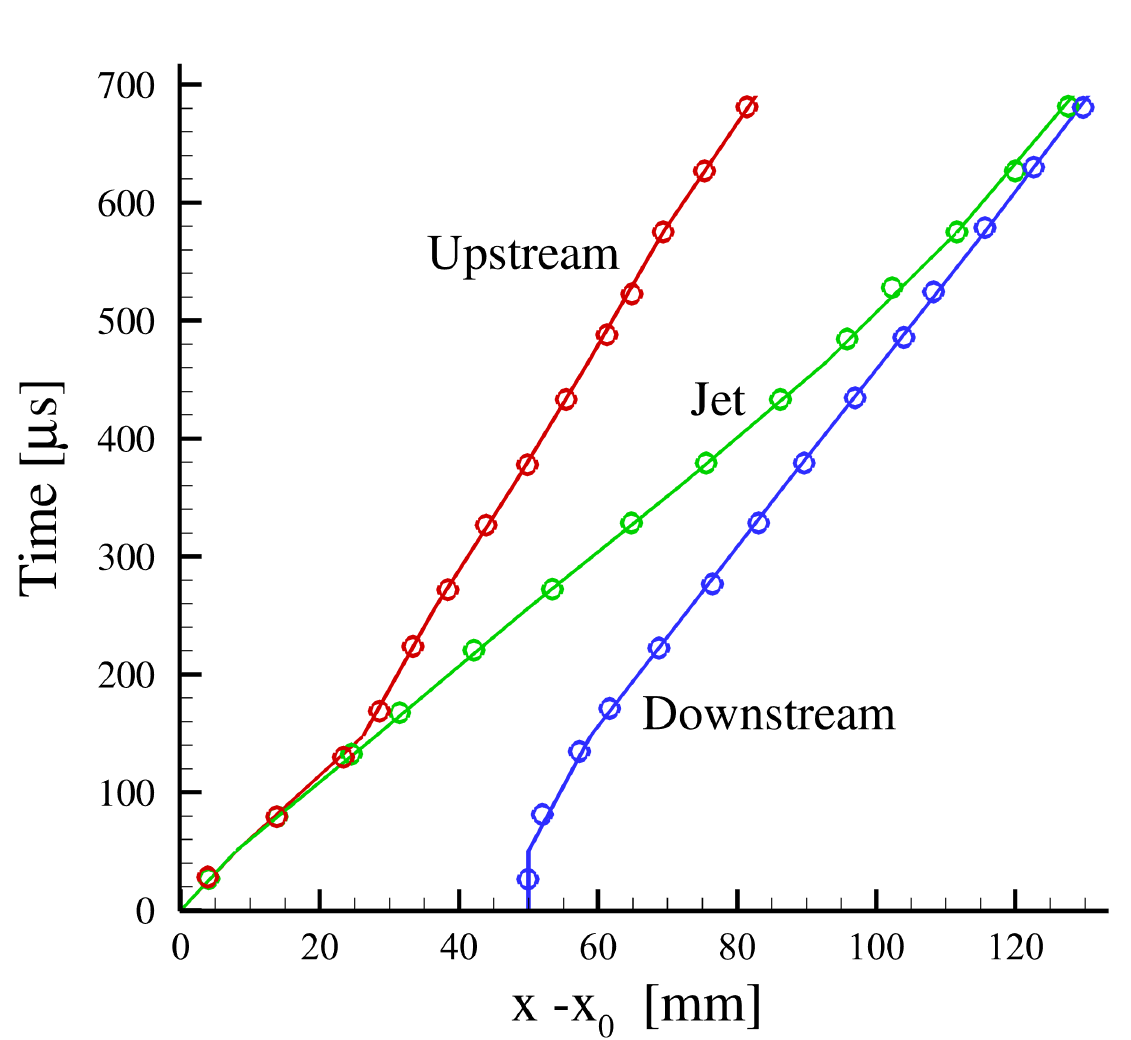} }}%
    \caption{Showing spatiotemporal evolution of characteristics points along the bubble interface from simulations (lines) with $h =25 \: \mu \text{m}$ and reported data by Terashima and Tryggvason (symbol circle) \cite{terashima2009front}.}%
    \label{fig:Jet_Up_Down}%
\end{figure}

\begin{figure}
    \centering
    \includegraphics[width=0.25\textwidth]{ 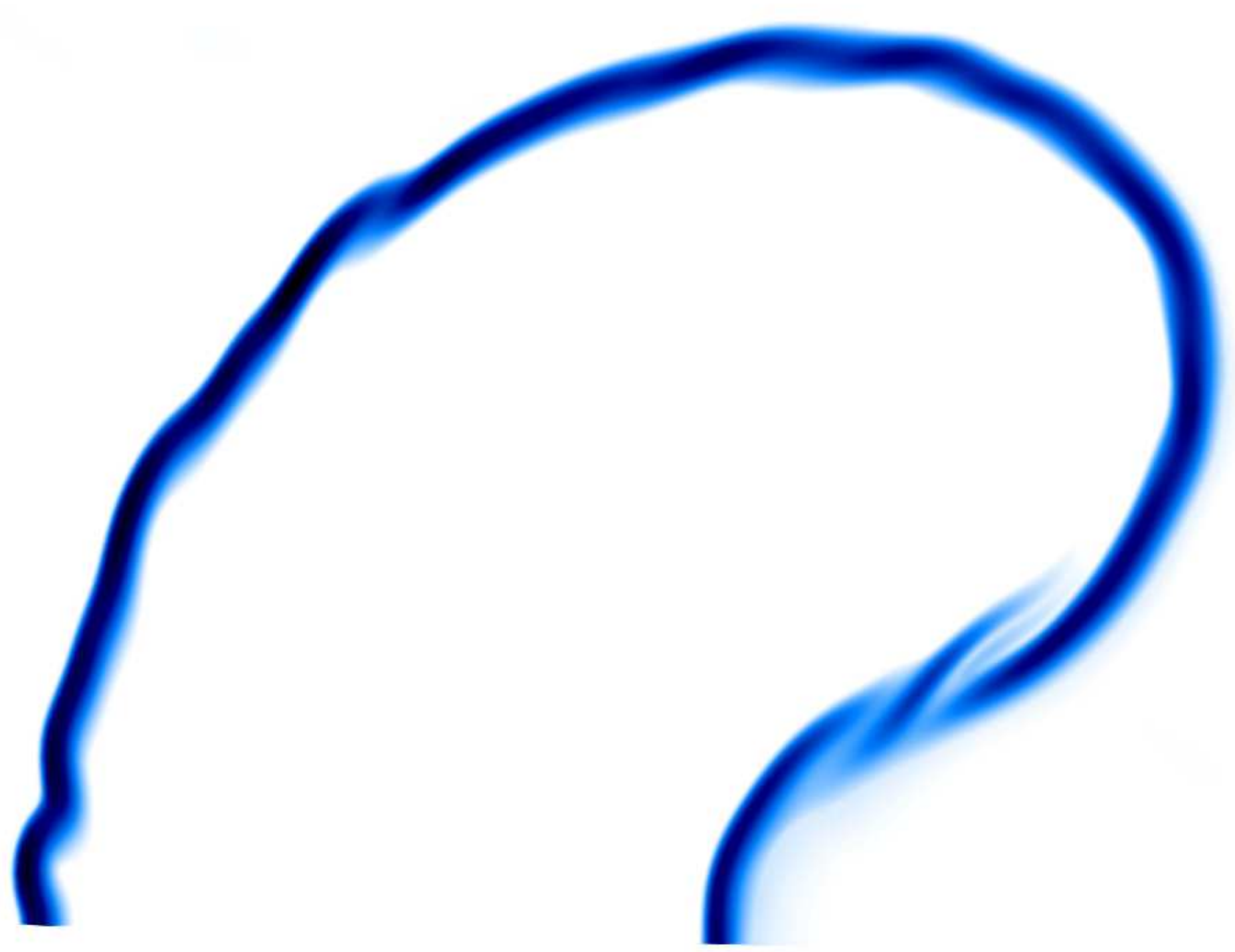}\hspace{0.45cm }    \includegraphics[width=0.25\textwidth]{ 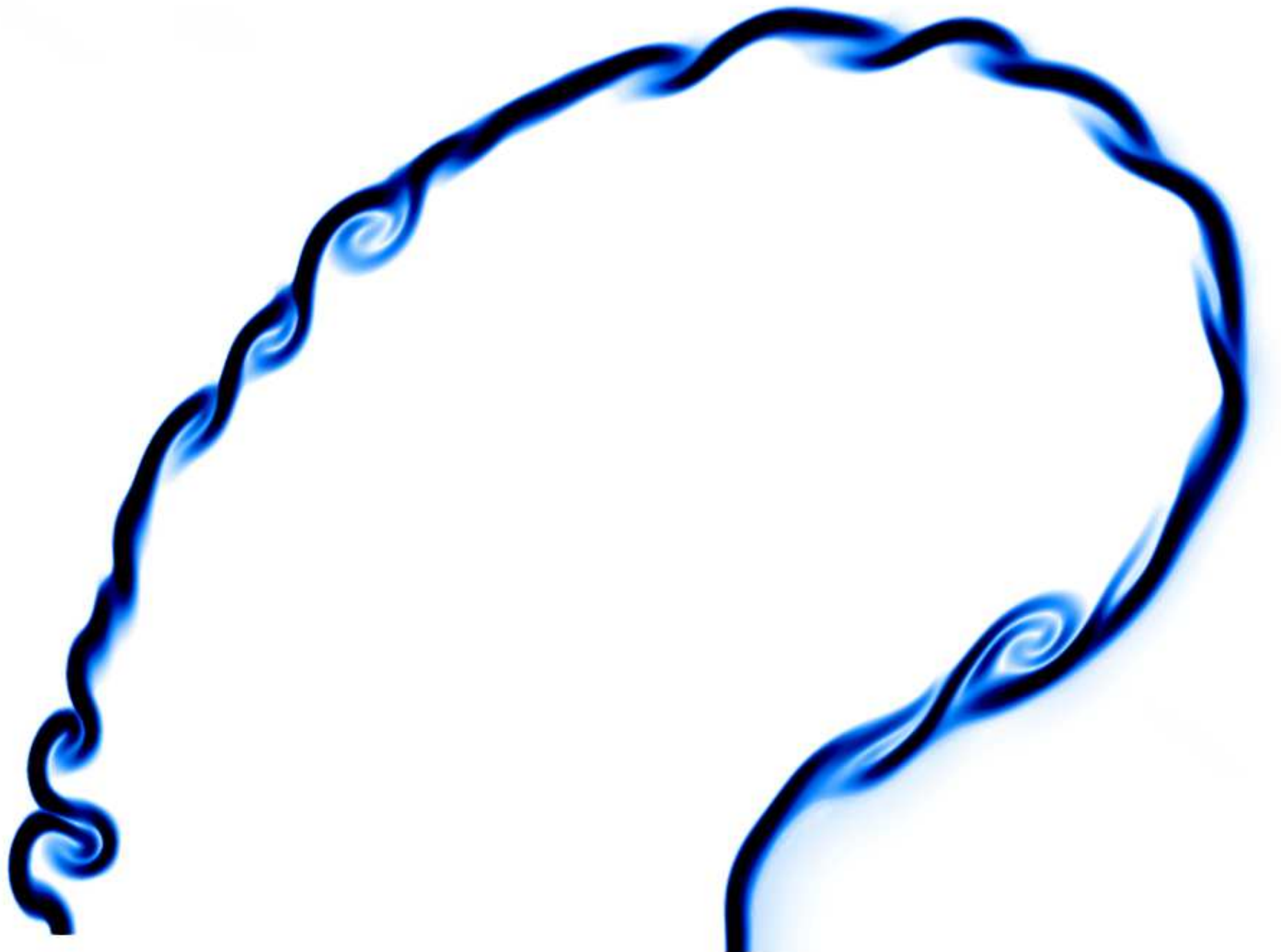}\hspace{0.45cm }    \includegraphics[width=0.25\textwidth]{ 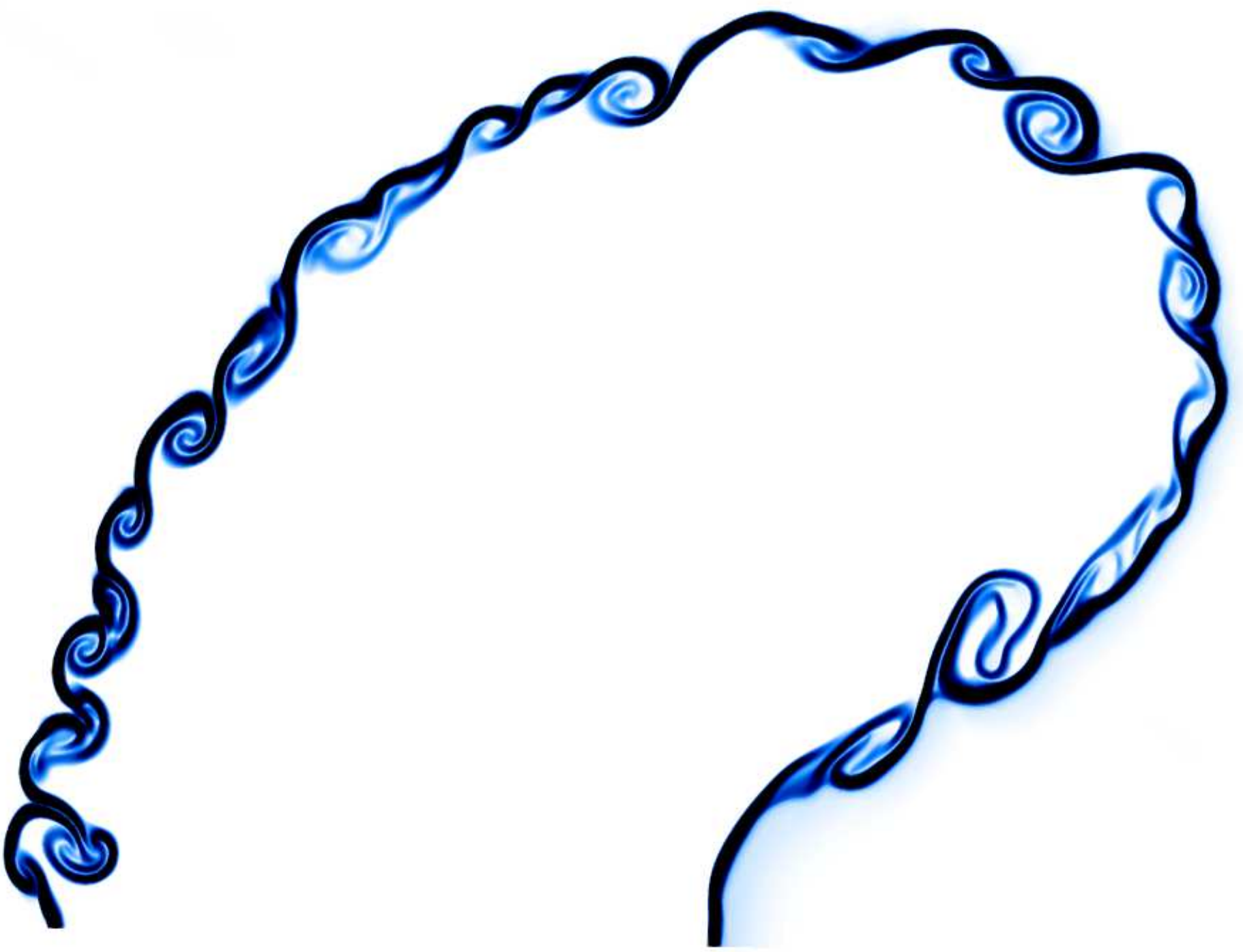}\\
    \includegraphics[width=0.25\textwidth]{ 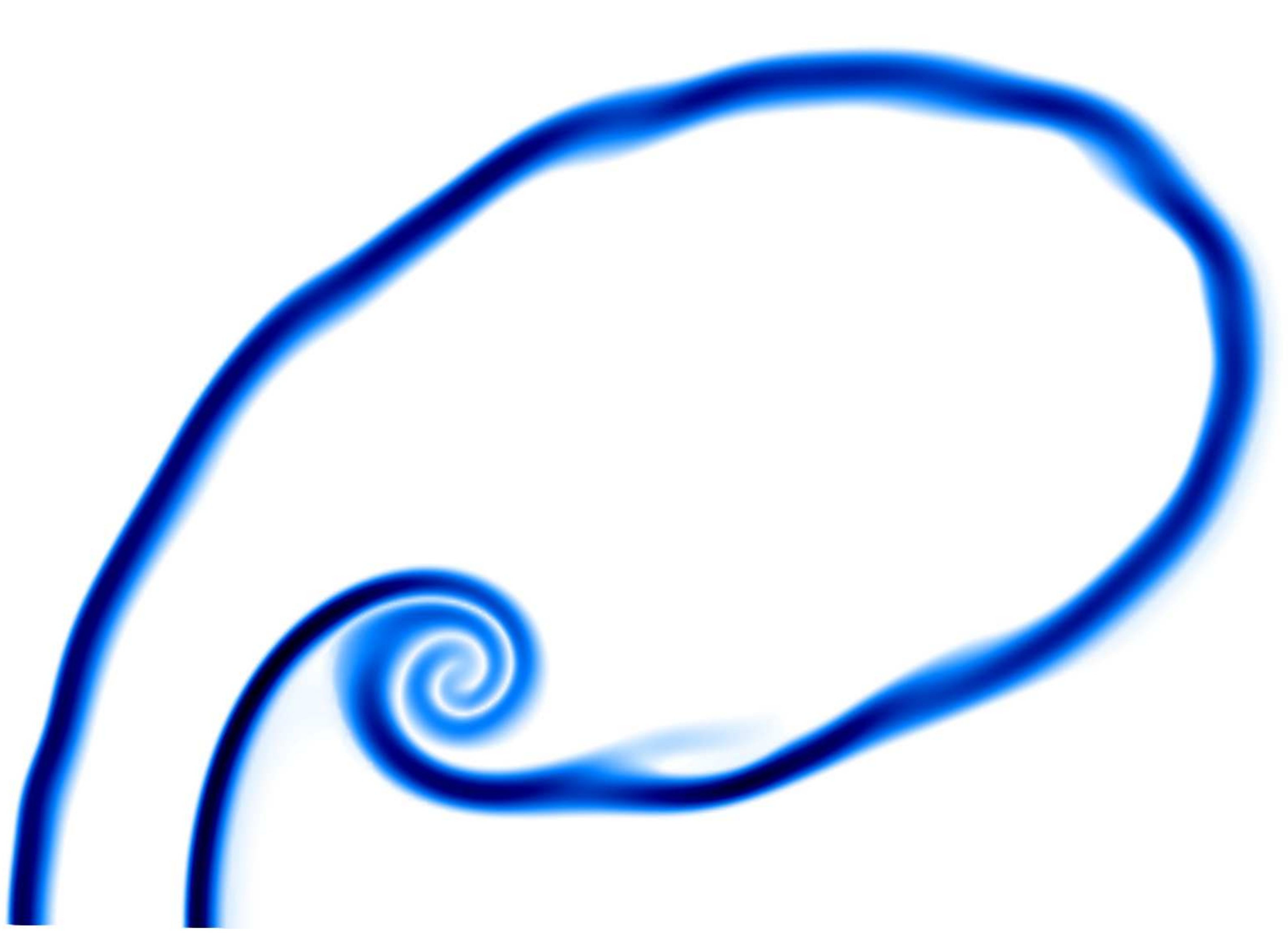}\hspace{0.45cm }    \includegraphics[width=0.25\textwidth]{ 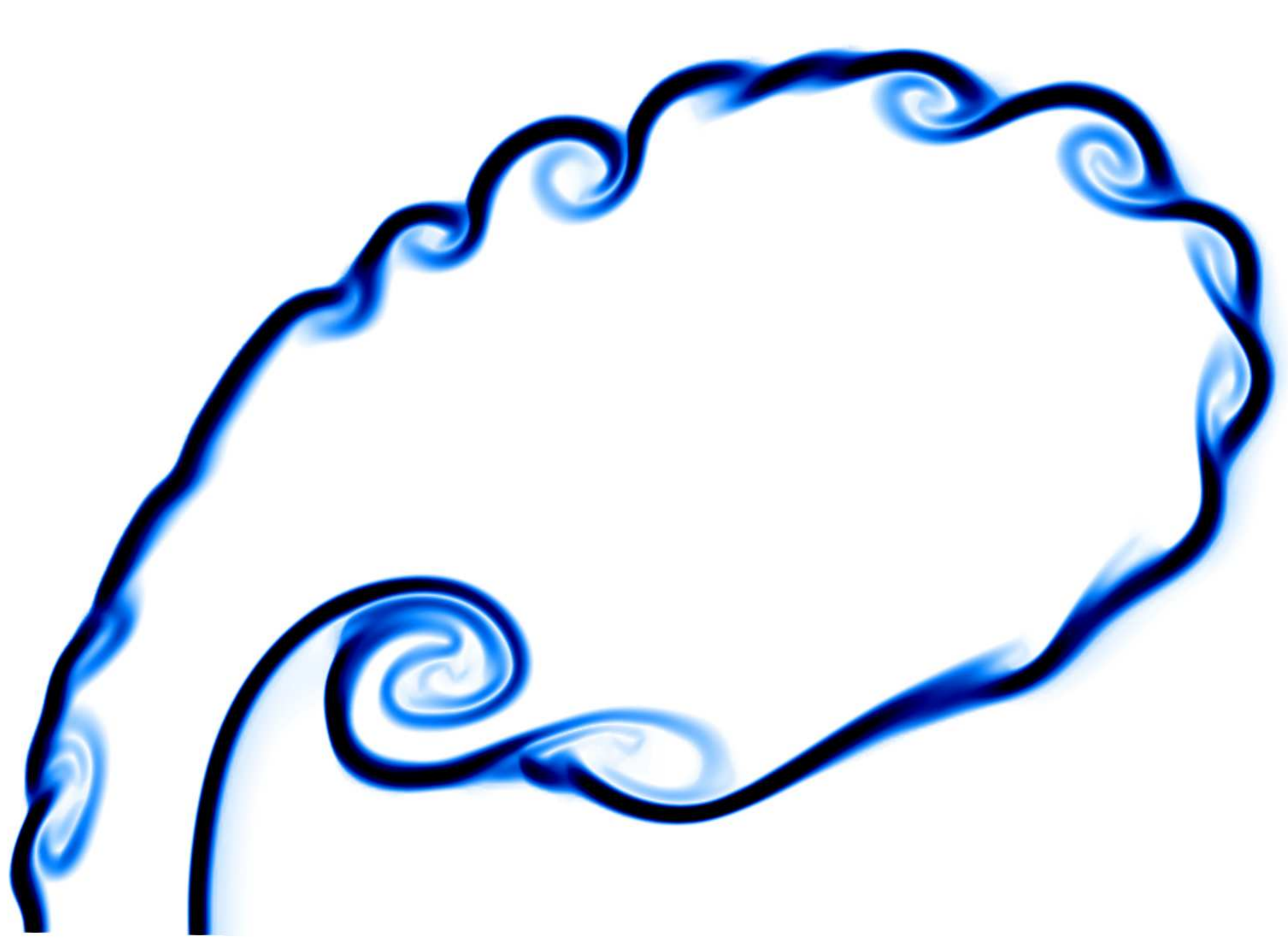}\hspace{0.45cm }    \includegraphics[width=0.25\textwidth]{ 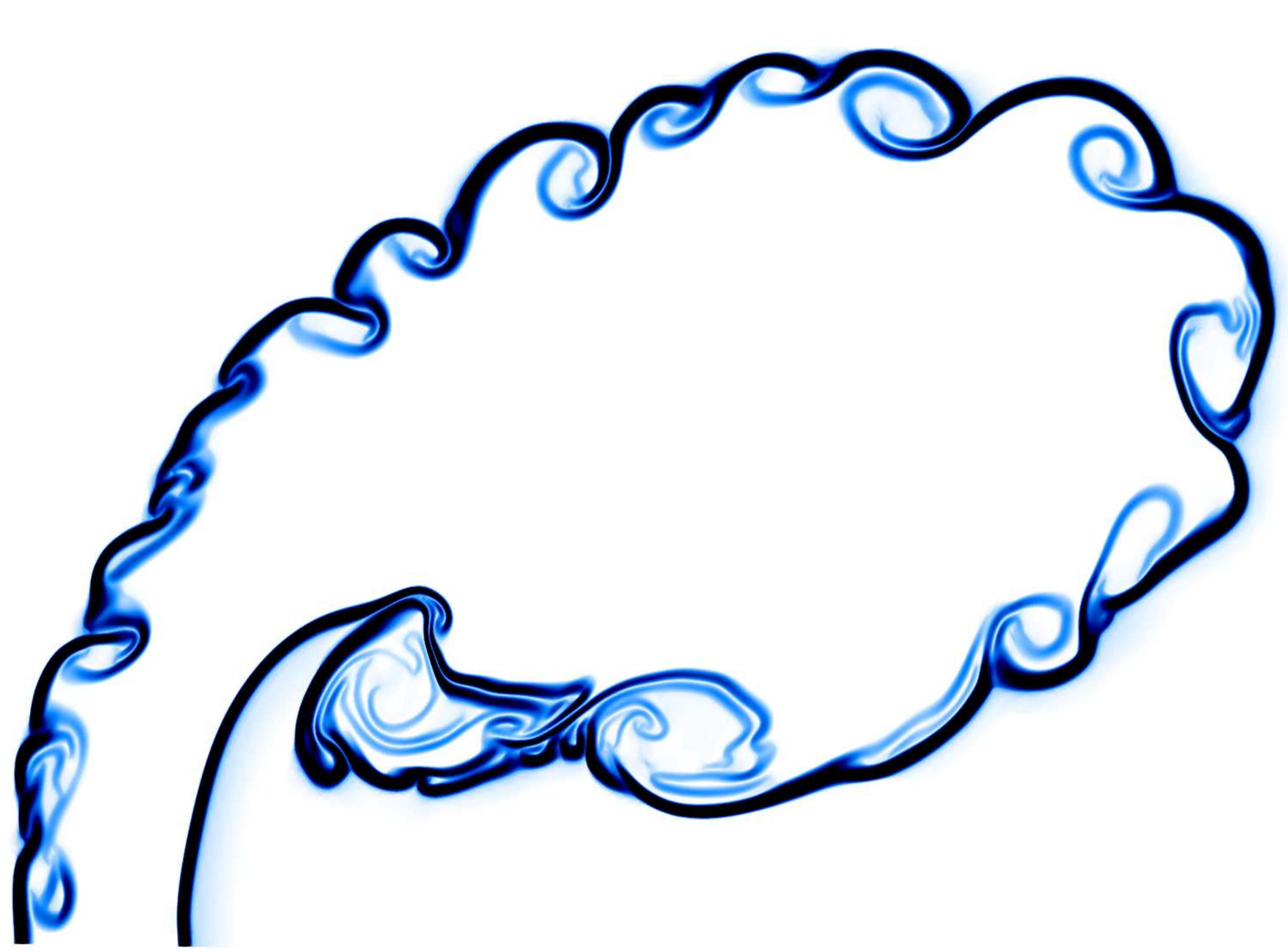}\\
    \subfloat[\centering Level 1, $h =100 \: \mu \text{m}$]{
    \includegraphics[width=0.25\textwidth]{ 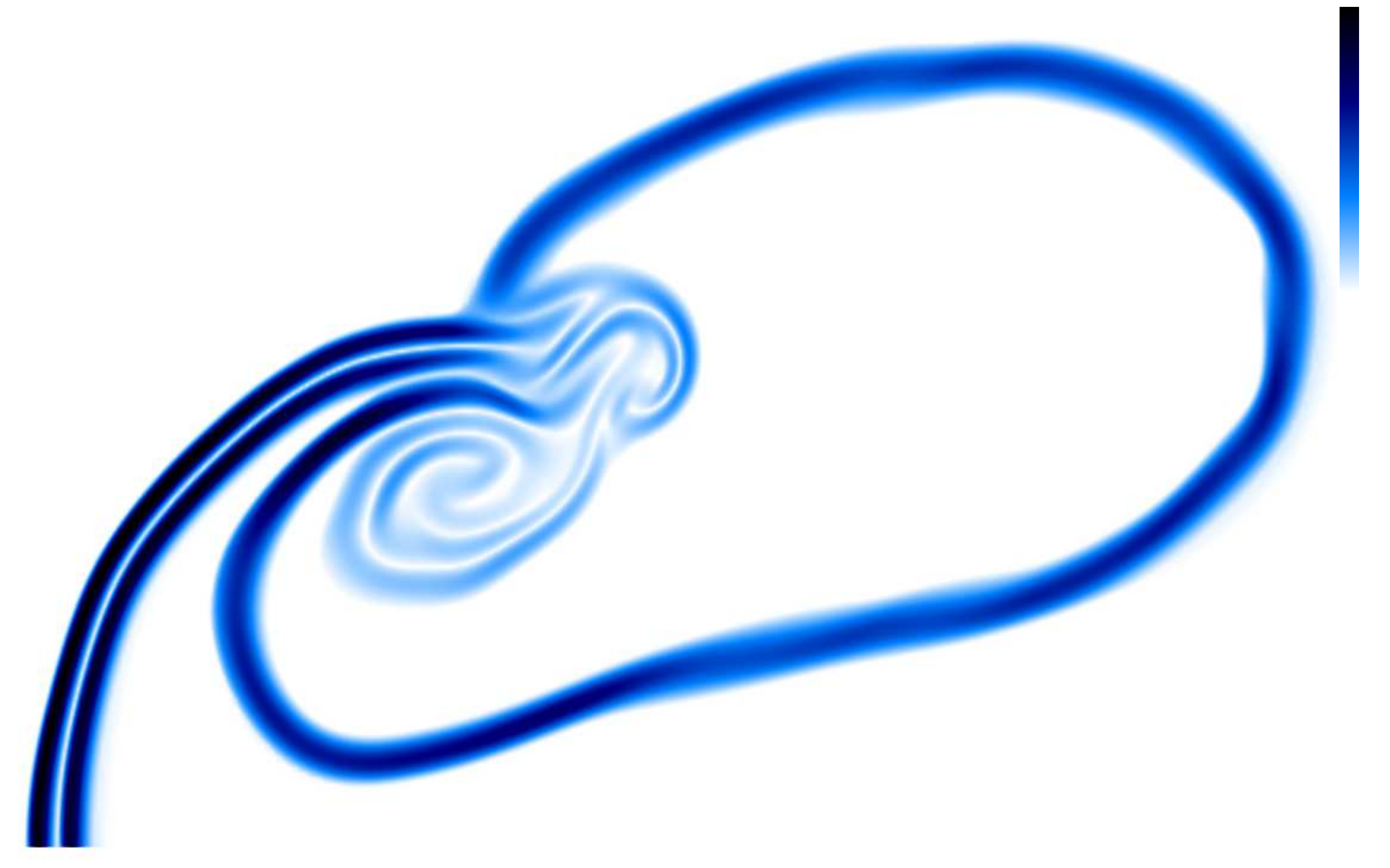}}\hspace{0.45cm}
    \subfloat[\centering Level 2, $h =50 \: \mu \text{m}$]{
    \includegraphics[width=0.25\textwidth]{ 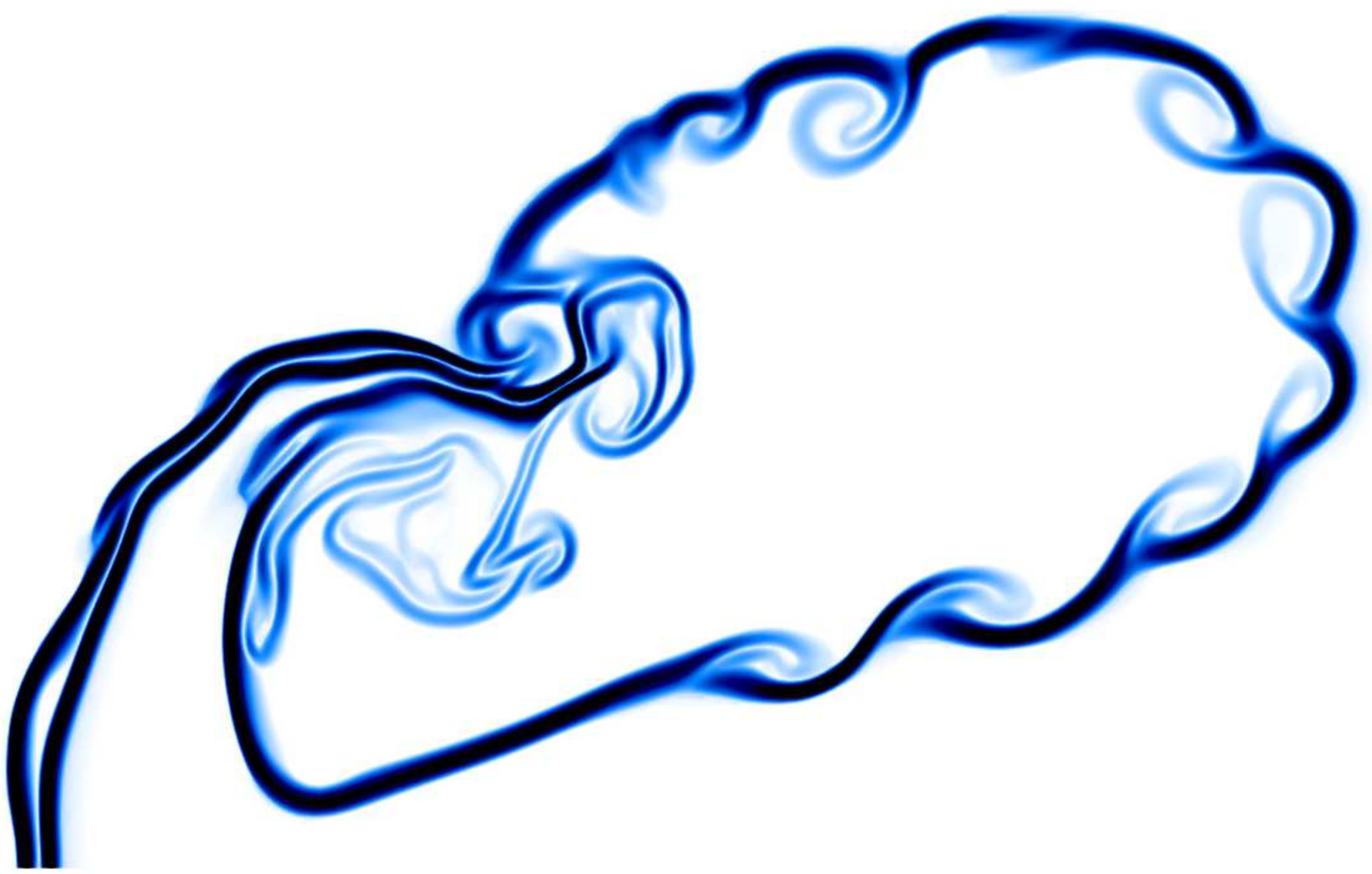}}\hspace{0.45cm } 
    \subfloat[\centering Level 3 , $h =25 \: \mu \text{m}$]{
    \includegraphics[width=0.25\textwidth]{ 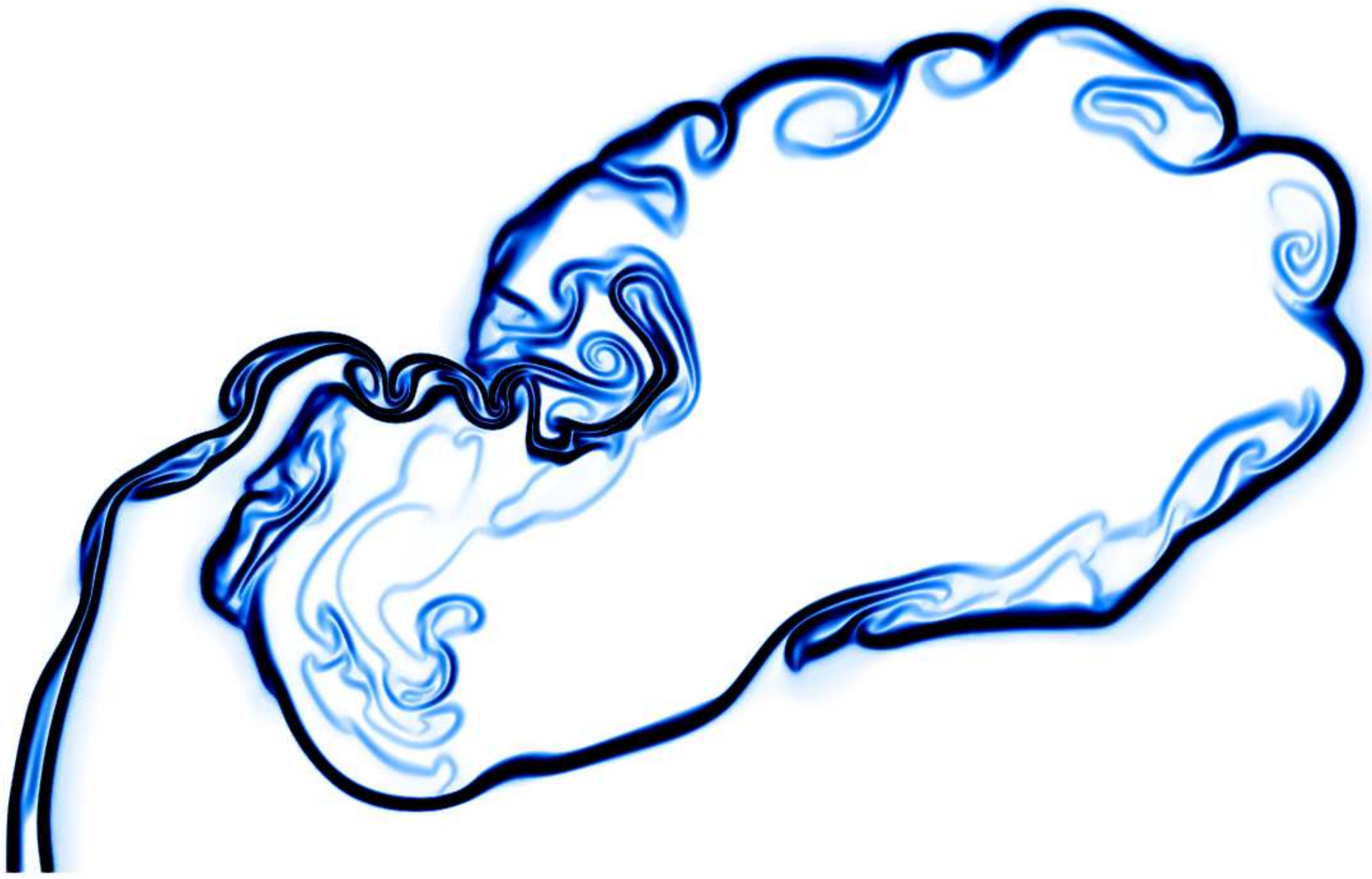}} 
    \caption{Sequence of local refinements for the two-dimensional Shock-bubble interaction problem. Time stamps from top to bottom: $t = 300 \: \mu \text{s}$, $t = 500 \: \mu \text{s}$, and $t = 700 \: \mu \text{s}$}
    \label{fig:local_refine}
\end{figure}

The influence of grid resolution on the density gradient at $t =300, 500$ and $ 700 \: \mu \text{s}$ is depicted in Figure~\ref{fig:local_refine}. Since this test involves fluid instabilities, achieving grid convergence is impossible until the computational cell size becomes significantly smaller than the Kolmogorov scale (approximately $2.5\: \mu \text{m}$). This remains impractical to achieve, even when utilizing AMR \cite{houim2011low}.
\clearpage
\subsection{Under-expanded hydrogen injection}  \label{res4}
The Entropy-Stable/Double-Flux scheme is evaluated through a supersonic hydrogen jet simulation. Hydrogen’s high-pressure injection creates shock waves, impacting fuel-air mixing and combustion. Simulating these flows is complex due to turbulence, acoustics, and shock interactions \cite{crist1966study,snedeker1971study,vuorinen2013large}.
To assess the numerical flux scheme’s behavior in multi-component jet flows, we adopted a benchmark similar to Hamzehloo et al. \cite{hamzehloo2014large} and conducted a simulation of an under-expanded hydrogen jet, following the experimental setup of Ruggles and Ekoto \cite{ruggles2012ignitability} with a nozzle pressure ratio (NPR) of 10. The simulation domain is a cylindrical region with a diameter of $30R_{in}$ and a length of $41R_{in}$, where $R_{in}= 1.5 \ $mm represents the nozzle diameter, see Figure \ref{fig:Jet_Testcase}.

We simplified the configuration by applying a Dirichlet inlet boundary condition instead of using a high-pressure reservoir. The inlet extends one nozzle diameter into the domain, while Euler slip walls are assigned to the surrounding boundaries, and supersonic outflow conditions are set for the remaining areas. The material parameters used in the three-dimensional under-expanded hydrogen injection simulation are provided in Table \ref{tab:initial_conditions_Jet}. The computational domain is structured using a structured hexahedral mesh, consisting of approximately 1.1 million cells at the start of the simulation. As a DNS of this configuration is unfeasible the Smagorinsky eddy viscosity model for the unresolved subgrid contribution is used to run the test case as a large eddy simulation (LES) to investigate the characteristics of hydrogen injection.

\begin{figure}
    \centering
    {{\includegraphics[width=0.99\textwidth]{ 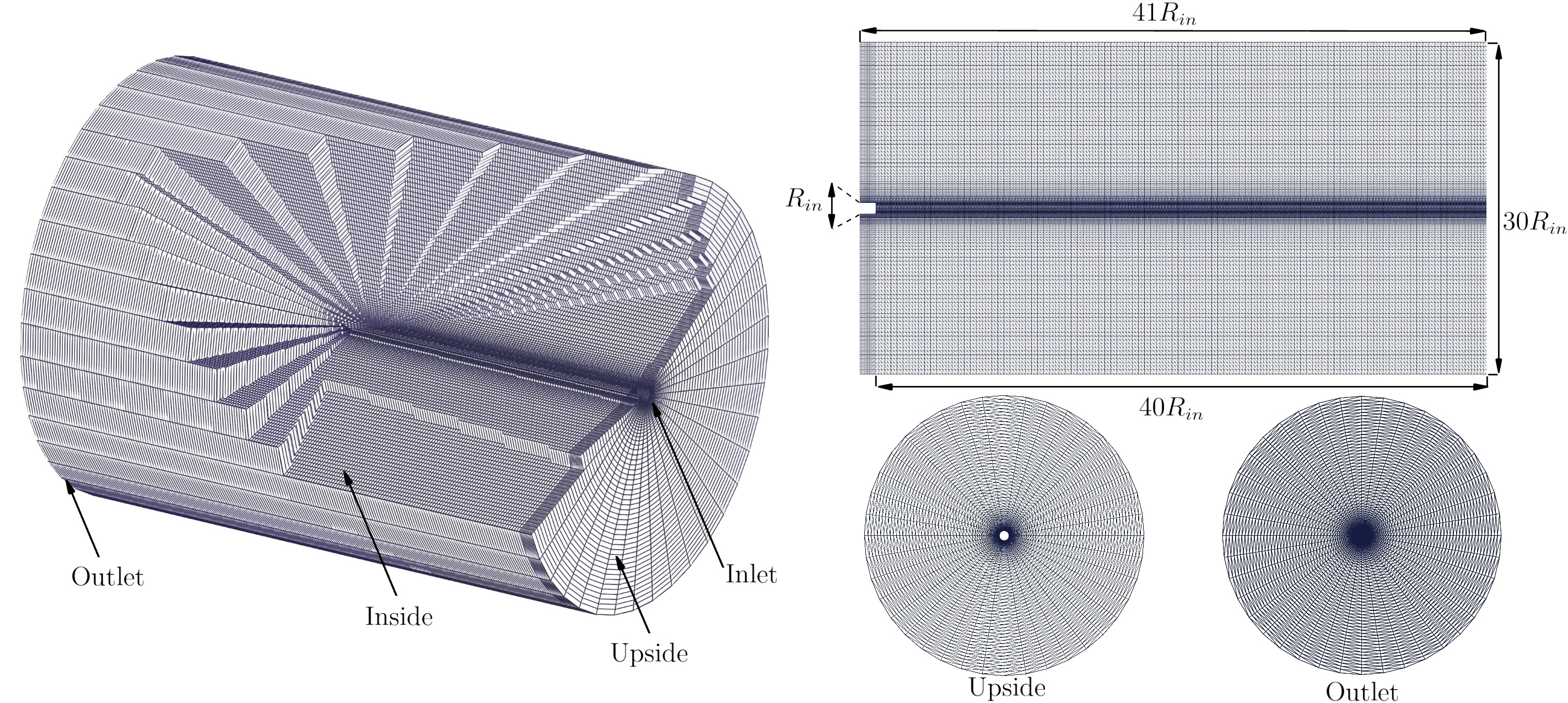} }}
    \caption{Computational domain for hydrogen injection. The right side presents the setup dimensions, while the left side visualizes a cut through the initial three-dimensional grid structure.}%
    \label{fig:Jet_Testcase}%
\end{figure}
\begin{table}[h]
\captionsetup{justification=raggedright, singlelinecheck=false}
    \caption{Initial conditions for the three-dimensional under-expanded Hydrogen injection problem.}
    \label{tab:initial_conditions_Jet}
\begin{tabular*}{\textwidth}{l@{\extracolsep{\fill}}lcc}
        \toprule
        Quantity   & Inlet & Inside  \\
        \midrule
        $\rho \: \:$ [kg/m$^3$]         & 0.5115    & 1.1579      \\
        $\text{v}_1$ [m/s]              & 1195.15   & 0.00        \\
        $\text{v}_2$ [m/s]              & 0         & 0            \\
        $\text{v}_3$ [m/s]              & 0         & 0            \\
        $p \:\:$ [bar]                  & 10        & 1            \\
        $T\:$ [K]                       & 245.6     & 296.0         \\
        $Y_{N_2}$                       & 0.000     & 0.215        \\
        $Y_{O_2}$                       & 0.000     & 0.785         \\
        $Y_{H_2}$                       & 1.000     & 0.000         \\        \bottomrule
    \end{tabular*}
\end{table}

\subsubsection{Simulation results and analysis}

The simulation runs until $t = 55 \: \mu \text{s}$ using $N_{\text{procs}} = 3840$ for both the Double-Flux with the HLLC flux function and Entropy-Stable/Double-Flux methods. To achieve higher resolution, we applied AMR with a second-layer local refinement, increasing the cell count by 33 million at $t = 55 \: \mu \text{s}$.  

The results align well qualitatively with those reported by \cite{hamzehloo2014large,ruggles2012ignitability}, as evident from the overall flow field comparison. This agreement is further supported by analyzing the position and size of the first Mach disk, as well as the jet penetration for NPR=10 at $t = 55 \: \mu \text{s} $, as presented in Table \ref{tab:var_Jet}, where a close qualitative match to the data in \cite{hamzehloo2014large,ruggles2012ignitability} is observed.

\begin{table}[h]
    \caption{Comparison of Mach disk dimensions and penetration depth for different studies and numerical methods in the three-dimensional under-expanded hydrogen injection problem for NPR=10 at $t = 55 \: \mu \text{s} $. AMR1 applies one refinement, while AMR2 applies two.}
    \label{tab:var_Jet}
\begin{tabular*}{\textwidth}{@{\extracolsep{\fill}}lccc}
        \toprule
        \multirow{2}{*}{Author(s)} & \multicolumn{2}{c}{Mach disk} & \multirow{2}{*}{Penetration [mm]}  \\
        \cmidrule(lr){2-3}
        & Height [mm] & Width [mm] &  \\
        \midrule
        Hamzehloo et al. \cite{hamzehloo2014large}          & 3.09   & 1.34 & 15.80      \\
        Ruggles and Ekoto \cite{ruggles2012ignitability}    & 3.05   & 1.30 & -        \\
        Double-Flux, AMR2                                   & 3.09   & 1.36 & 15.86         \\
        Entropy-Stable/Double-Flux, AMR1                    & 3.14   & 1.38 & 15.78         \\
        Entropy-Stable/Double-Flux, AMR2                    & 3.10   & 1.32 & 15.84         \\        
        \bottomrule
    \end{tabular*}
\end{table}

The initial transient stages of near-nozzle shock expansion, Mach disk formation, and jet development in the under-expanded hydrogen flow are shown in Figure \ref{fig:Hydrogen injection}. We can comparing two methods at time $t =25 \: \mu \text{s}$ and $t =55 \: \mu \text{s}$. Both methods predict a similar Mach disk position, indicating consistency in capturing the primary shock structures. However, the Entropy-Stable/Double-Flux method reduces numerical dissipation, allowing for a more accurate representation of flow features. In particular, it successfully captures the oblique shock with greater clarity, but with Double-Flux, oblique shock is very diffusive, especially at $t= 25 \: \mu \text{s}$. This behavior is clear for jet boundary and emitted sound in both times.

Hamzehloo et al. \cite{hamzehloo2014large} demonstrated that mixing occurs before the Mach disk. Where with hydrogen, high levels of momentum exchange and mixing were observed at the boundary of the intercepting shock. The high turbulence fluctuations at the nozzle exit of the hydrogen jet triggered Görtler vortices \cite{hamzehloo2014large}. As shown in Figure \ref{fig:Hydrogen injection}, at $t = 55 \: \mu \text{s}$, mixing did not occur before the Mach disk for the Double-Flux method. In contrast, the Entropy-Stable/Double-Flux method exhibited mixing before the Mach disk. For both methods, primary mixing was observed to take place after the Mach disk, particularly near the jet boundaries, where large-scale turbulence played a dominant role \cite{hamzehloo2014large}.

In addition, Figure \ref{fig:Hydrogen injection_velocity} presents a velocity field slice at $t = 55 \: \mu \text{s}$, comparing both methods. As shown in Figure \ref{fig:Hydrogen injection_velocity}, both methods capture shock waves, consistent with the findings of \cite{hamzehloo2014large,ruggles2012ignitability}. However, the Entropy-Stable/Double-Flux method demonstrates greater stability and remains free of oscillations beyond the Mach disk position. This improved stability is attributed to its enforcement of the second law of thermodynamics.

\begin{figure}
    \centering
    {{\includegraphics[width=0.9\textwidth]{ 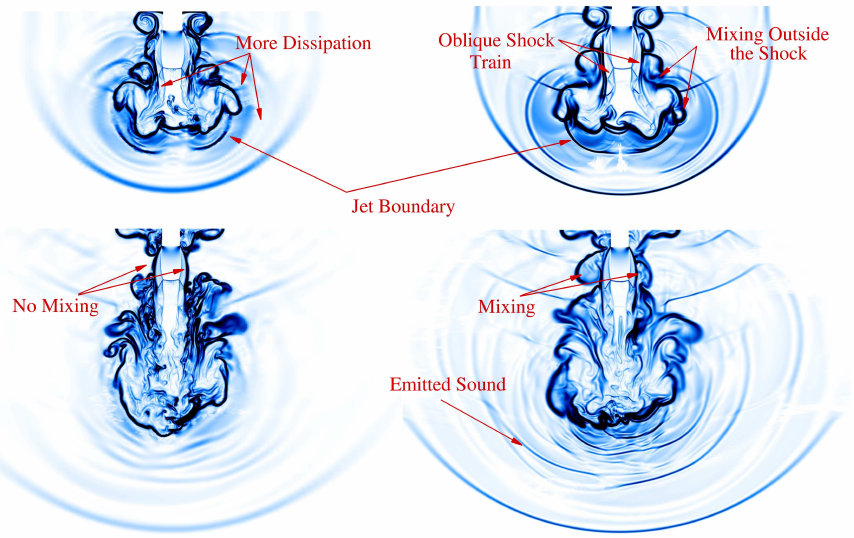} }}\\
    \caption{Comparison between the density gradient $\vert \nabla  \rho \vert $ of the Double-Flux scheme (left) and the Entropy-Stable/Double-Flux method (right). Time stamps from top to bottom: $t = 25 \: \mu \text{s}$ and $t = 55 \: \mu \text{s}$.}%
    \label{fig:Hydrogen injection}%
\end{figure}

\begin{figure}
    \centering
    {{\includegraphics[width=0.95\textwidth]{ 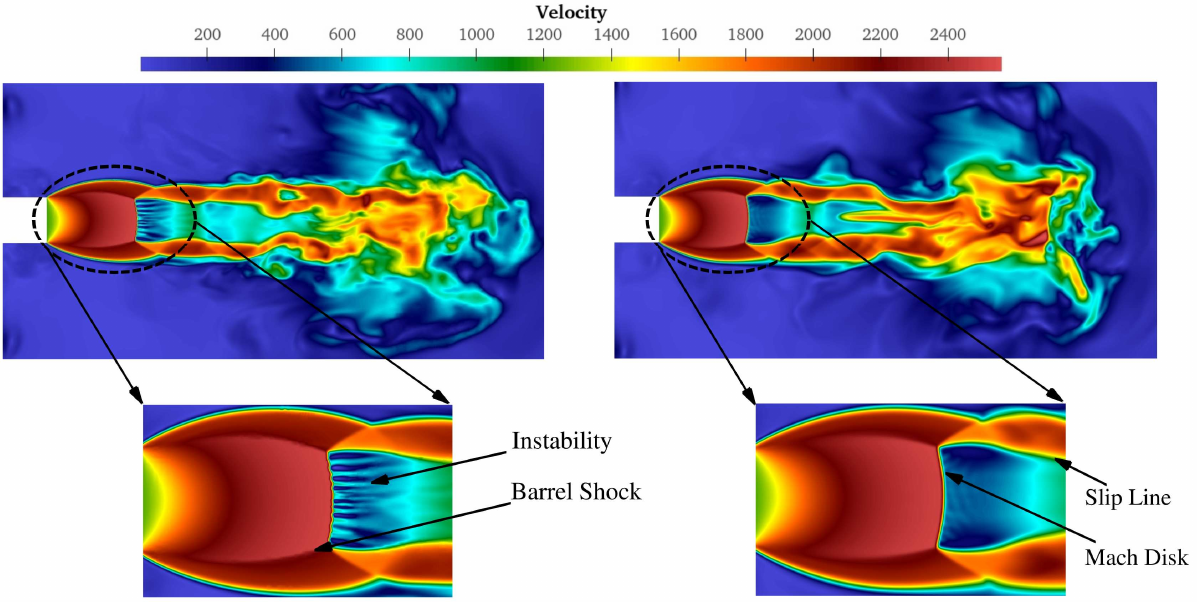} }}
    \caption{Two-dimension slice of the velocity field for under-expanded hydrogen injection into air at $t = 55 \: \mu \text{s}$, comparing the Double-Flux scheme (left) and the Entropy-Stable/Double-Flux method (right).}%
    \label{fig:Hydrogen injection_velocity}%
\end{figure}

\section{Conclusions and future work\label{sec:Conclusions and future work}}
This study addresses the numerical challenges associated with simulating multi-component compressible flows that involve strong species mixing, shock interactions, and turbulent dynamics. Such flows often give rise to non-physical oscillations in pressure and temperature due to discontinuities in molar mass and specific heat ratios. Moreover, ensuring consistency with the second law of thermodynamics and maintaining numerical stability in the presence of sharp gradients remains a significant challenge in conventional numerical frameworks.

To overcome these limitations, we have developed a new numerical flux formulation that combines the Entropy-Stable method with the Double-Flux approach within a finite volume discretization. This formulation is further enhanced by the inclusion of a hybrid dissipation mechanism, which improves robustness without compromising accuracy. The resulting method ensures entropy consistency, kinetic energy conservation, and oscillation-free behavior across species' interfaces. It also provides enhanced shock-capturing capabilities, particularly in flows with strong compressibility and mixing.

The proposed method was implemented in a density-based solver within the OpenFOAM environment. Its effectiveness was evaluated through a series of progressively challenging test cases designed to assess accuracy, stability, and physical fidelity.

The first test case examines a one-dimensional, single-component isentropic flow with a discontinuous initial profile. In this simplified scenario, the Entropy-Stable/Double-Flux formulation is adapted into an entropy-conserving scheme, demonstrating optimal convergence in entropy variation and confirming the method's theoretical consistency.

The second case considers a two-dimensional, multi-component flow with a moving material interface involving two distinct species. Conventional Entropy-Stable fluxes typically produce spurious oscillations in pressure and velocity at such interfaces. However, the proposed numerical flux successfully eliminates these artifacts, capturing the interface with high fidelity and without introducing non-physical behavior.

The third test case involves a two-dimensional shock–bubble interaction in a multi-component setting. This highly dynamic problem presents strong gradients and interactions between shocks and material interfaces. The Entropy-Stable/Double-Flux scheme closely matches reference solutions and demonstrates improved robustness and numerical stability, even under severe flow conditions.

In more complex, real-world scenarios, such as the simulation of a supersonic under-expanded hydrogen jet, the method shows an enhanced capability to resolve intricate shock structures, including Mach disks and barrel shocks, while maintaining stability and accuracy downstream of the shock front. Also, the simulation outcomes show strong agreement with both experimental observations and numerical benchmarks from the literature.

In summary, the Entropy-Stable/Double-Flux method offers a significant advancement in the numerical simulation of multi-component compressible flows, especially those involving strong species mixing, shocks, and turbulence. By unifying Entropy-Stable fluxes with Double-Flux mechanisms and augmenting them with a hybrid dissipation model, our method provides a robust, accurate, and thermodynamically consistent framework—paving the way for large-scale simulations, such as LES of turbulent reactive flows, in future work.

\section*{CRediT authorship contribution statement}
\textbf{Vahid Badrkhani:} Writing – review \& editing, Writing – original draft, Visualization, Validation, Software, Methodology, Investigation, Formal analysis, Data curation, Conceptualization. \textbf{T. Jeremy P. Karpowski:} Writing – review \& editing, Software, Data curation. \textbf{Christian Hasse:} Writing – review \& editing, Project administration, Funding acquisition, Conceptualization.

\section*{Declaration of competing interest}
The authors declare that they have no known competing financial interests or personal relationships that could have appeared to influence the work reported in this paper.

\section*{Acknowledgements}
This work has been funded by the Federal Ministry of Education and Research of the Federal Republic of Germany (BMBF) through the collaborative project CFD4H2, funding reference 03SF0640B, and the German Research Foundation (DFG) – Project no. 349537577. 

\section*{Data availability}
The data that support the findings of this study are available from the corresponding author upon reasonable request.
 \pagebreak
 
 \bibliographystyle{elsarticle-num}
\bibliography{sample}

\end{document}